\documentclass[draftcls,onecolumn,peerreview,letterpaper]{IEEEtran}
\usepackage{amsmath}
\usepackage{amssymb}
\usepackage{amsbsy}
\usepackage[dvips]{graphicx}
\newtheorem{prop}{Proposition}
\newtheorem{corr}{Corollary}

\newtheorem{rem}{Remark}
\newtheorem{lem}{Lemma}
\newcommand{\eps}{\epsilon}

\newcommand {\dotle} {\stackrel{.} {\le}}  
\newcommand {\dotge} {\stackrel{.} {\ge}}  

\begin{document}
\title{Distributed MIMO receiver - Achievable rates and upper bounds}

\author{\authorblockN{Amichai Sanderovich,Shlomo Shamai (Shitz),Yossef Steinberg}\\
\authorblockA{Technion, Haifa, Israel}
 }
\maketitle
\begin{abstract}
In this paper we investigate the achievable rate of a system that
includes a nomadic transmitter with several antennas, which is
received by multiple agents, exhibiting independent channel gains
and additive circular-symmetric complex Gaussian noise. In the
nomadic regime, we assume that the agents do not have any decoding
ability. These agents process their channel observations and forward
them to the final destination through lossless links with a fixed
capacity. We propose new achievable rates based on elementary
compression and also on a Wyner-Ziv (CEO-like) processing, for both
fast fading and block fading channels, as well as for general
discrete channels. The simpler two agents scheme is solved, up to an
implicit equation with a single variable. Limiting the nomadic
transmitter to a circular-symmetric complex Gaussian signalling, new
upper bounds are derived for both fast and block fading, based on
the vector version of the entropy power inequality. These bounds are
then compared to the achievable rates in several extreme scenarios.
The asymptotic setting with numbers of agents and transmitter's
antennas taken to infinity is analyzed. In addition, the upper
bounds are analytically shown to be tight in several examples, while
numerical calculations reveal a rather small gap in a finite
$2\times2$ setting. The advantage of the Wyner-Ziv approach over
elementary compression is shown where only the former can achieve
the full diversity-multiplexing tradeoff. We also consider the
non-nomadic setting, with agents that can decode. Here we give an
achievable rate, over fast fading channel, which combines broadcast
with dirty paper coding and the decentralized reception, which was
introduced for the nomadic setting.

\end{abstract}
\section{Introduction} In this paper we deal with a network in which a nomadic transmitter has several antennas and is
communicating to a remote destination, where no direct link exists
between the transmitter and the final destination, as is depicted in
figure \ref{fig:setting}. The final destination receives all of its
inputs from several separated agents, which are connected to it
through fixed lossless links with a given capacity. This setting is
identical to the setting of \cite{fullpaper}, only that here we
focus on fading channels. Namely, the channel between the
transmitting antennas and the agents is a Rayleigh fading channel
with independent channel gains, where the extension to other fading
statistics is straight forward. In this contribution we consider
both fast fading and block fading channels. The channel fading
coefficients, or channel state information (CSI) are known in full
to the agents and the final destination, but not to the transmitter.
This setting is closely related to the setting of the multiple input
multiple output (MIMO) channel, which is thoroughly treated in the
literature, see \cite{Telatar99}. The multiplexing gain is a common
asymptotic measure of performance of MIMO systems. It assesses the
capacity increase, for high signal to noise ratios, due to the use
of multiple antennas \cite{TseZheng2003} in the scheme. In this
paper, we analyze the multiplexing gain for the suggested network,
where recent examples for the multiplexing gains of multi terminal
networks are \cite{AvestimehrTse2006} and \cite{YukselErkip2006}.
The results reported here have implications on other MIMO-related
channels, such as the MIMO broadcast channel
\cite{WeingartenSteinbergShamai2006}, the MIMO relay channel
\cite{WangZhangHostMadsen2005}, and ad-hoc networks
\cite{BolcskeiNabar2004}. All these works deal with situations where
multiple antennas are transmitting and the signals are received in a
distributed fashion, either by relays, destinations or any
combination of the above. In addition, results regarding ad-hoc
networks \cite{GuptaKumar2003}, relay channels \cite{MaricYates2004}
and joint cell-site processing \cite{SomekhZaidelShamai2004} are
closely related, providing yet another aspect of the achievable
rates in wireless networks, where relays form, in a distributed
manner, the required spacial dimensions.

Our model assumes that the transmitter is nomadic, which means that
the agents do not possess the codebook in use, and thus do not have
any decoding ability \cite{fullpaper}. A good way to model a nomadic
setting is by letting the transmitter use random encoding. Such
model excludes any decoding from the agents. Given that the codebook
is random, we further assume that it is Gaussian. In this case, as
the model becomes close to source coding and the Gaussian CEO (Chief
Executive Officer) \cite{Oohama1998}, we were able to obtain
analytic expressions for an achievable rate and for upper
bounds. 
Relevant works here are distributed source coding by Wyner and Ziv
(WZ) \cite{WynerZiv1976}, \cite{wagner-2005-} who deals with the
multiple terminals WZ problem and the Gaussian CEO by
\cite{Oohama2005}, among others.

The achievable rates derived in this paper extend the achievable
rates from \cite{fullpaper} to the case of fading channels, and
multiple antennas at the transmitter and at the receiver. The
techniques that are used for the derivation are based on the well
known CEO or WZ distributed source coding. These techniques,
although intended for source coding problems, enable better
utilization of system resources also for channel coding problems,
as done for example by
\cite{DraperWornell2004},\cite{CoverElgamal1979} and
\cite{KramerGastparGupta2004}.

The upper bounds in this paper were derived using the vector version
of the entropy power inequality, which was used for several known
problems which are based on Gaussian statistics. These include the
MIMO broadcast channel \cite{WeingartenSteinbergShamai2006} and the
Gaussian CEO with quadratic distortion \cite{Oohama2005}. Several
generalizations to the original entropy power inequality exist,
among them are \cite{LiuViswanath2006},\cite{WitsenhausenWyner1975}
and \cite{ChayatShamai1989}.

The Gaussian signaling used by the transmitter results with the
channel outputs being Gaussian and for the nomadic setting, also
memoryless. Notice that unlike traditional source coding problems
that use the CEO or WZ techniques, and examine the resulting
distortions, we focus on the allowed communication rates. Thus any
upper bounds or even optimality shown for a source coding problem,
although strongly connected, is not identical to our problem.
Therefore the technique used to show optimality of the distributed
WZ with two terminals problem (\cite{wagner-2005-}) does not carry
over to our setting.

This paper is organized as follows, in section \ref{sec:setting} the
setting is described and the basic definitions and notations are
given. Section \ref{sec:simple} describes the elementary compression
approach and gives several results about the achievable rates when
using this approach. Section \ref{sec:WynerZiv} improves upon the
approach taken in section \ref{sec:simple} by including CEO
compression (as in the CEO problem) at the agents and the final
destination. An upper bound to the achievable rate, when using
nomadic transmitter and non-decoding agents, is given in section
\ref{sec:upperbound}, and then demonstrated by a numerical example,
to be rather close to the achievable rate when using the CEO
compression. In the last section, an achievable rate for when the
agents are informed of the codes used by the transmitter, and the
transmitter is informed of both agents' processing and channel
coefficients is given in section \ref{sec:cogniz}. Concluding
remarks are then made in section \ref{sec:conclusion}.

\section{Setting and Model definition}\label{sec:setting}
Throughout this paper, boldfaced letters are used to denote
vectors $\boldsymbol{X}$ of length $n$, calligraphic letters
$\mathcal{T}$ to denote sets, capital letters $X$ are usually used
for random variables, and lower case letters for realizations of
random variables $x$, indices $i,j,k$, and counters $n,r,t$.
Subscript denotes an element within a vector and superscript $X^r$
denotes the set $X_1,\dots,X_r$.

The statistical mean is denoted by $\mathrm{E}$, ${}^*$ denotes
the transpose conjugate and $\mathcal{CN}(\Xi,\Sigma)$ stands for
complex Gaussian random variable with mean $\Xi$ and covariance
$\Sigma$.

An example for the model appears in Figure \ref{fig:setting}. The
model consists of a transmitter $S$ which has $t$ transmitting
antennas and which transmits during $n$ channel uses. In each
channel use, the transmitter sends a vector $X\in \mathbb{C}^{[t
\times 1]}$ to the channel, where
$\frac{1}{n}\sum_{k=1}^n\mathrm{E}[X(k)^*X(k)]\leq P$. 
The transmitter uses circular-symmetric complex Gaussian signalling,
which is known to be optimal for various problems involving the
Gaussian channel. The communication rate is denoted by $R$. The
message to be sent $M$ is encoded by a random encoding function
\mbox{$\boldsymbol{X}=\phi_{S,F}(M)$} such that for all messages
$M$, the outputs of the encoding function are randomly and
independently chosen according to probability
$P_{\boldsymbol{X}}(\boldsymbol{x})$.
We indicate the random encoding function by a random variable $F$.
That is,
\begin{equation}
\phi_{S,F}: [1,\dots,2^{nR}]\rightarrow\mathcal{X}^n.
\end{equation}
The agents are not informed about the selected encoding $F$, but are
fully aware of $P_X$.

We have $r$ agents $A_1,\dots,A_r$, each receiving the scalar
channel outputs
\begin{equation}\label{eq:def}
Y_i(k)=H_i(k) X(k) + N_i(k),\ \ i=1,\dots,r,\ k=1,\dots,n
\end{equation}
where $H_i(k)\in\mathbb{C}^{[1\times t]}$ is the vector of the
channel transfer coefficients, which are either ergodic (fast
fading) or static, non-ergodic (block fading). In both cases, the
coefficients are distributed independently from each other, and from
any other variable, according to circular-symmetric complex Gaussian
distribution $\mathcal{CN}(0,1)$. Similarly, the noises are
distributed as $N_i(k) \sim \mathcal{CN}(0,1)$, and are independent
of each other and along time. For the sake of brevity, we drop the
time index $k$ from now on.

Most of the results which are reported here can be easily extended
by including other fading distributions, such as Ricean, invoking
the results of \cite{TulinoLozanoVerdu2006}.

The $r$ agents are connected to a remote destination $D$ with
lossless links, each with capacity $C_i$ bits per channel use. The
final destination $D$ decodes the message $M$ from the $r$ messages,
which are sent from the $r$ agents, where decoding function is
$\phi_{D,F}:[1,\dots,2^{\sum^r nC_i}]\rightarrow [1,\dots,2^{nR}]$.
This setting is depicted in figure \ref{fig:setting}.

For fast fading channels, the rate $R$ is said to be achievable, if
for every $\eps>0$, there exists $n$ sufficiently large such that
\begin{equation}
\frac{1}{2^{nR}}\sum_{m=1}^{2^{nR}}\Pr(\hat{M}\neq m|M=m)\leq \eps,
\end{equation}
where $\Pr(\hat{M}\neq m|M=m)$ includes averaging over the channel
and the random coding. In parallel, the rate-vs-outage probability
of $\eps$, for block fading is said to be achievable if there exists
$n$ sufficiently large such that
\begin{equation}
\frac{1}{2^{nR}}\sum_{m=1}^{2^{nR}}\Pr(\hat{M}\neq m|M=m)\leq \eps.
\end{equation}

\begin{figure}
\centering
 \includegraphics[width=3in]{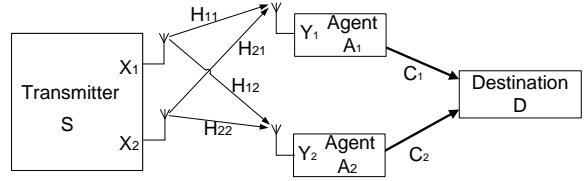}
 \caption{A system that includes a transmitter with $t=2$ and two agents $A_1$ and $A_2$ ($r=2$),
 connected to the final destination with capacities of $C_1$ and $C_2$, respectively.
 The channel fading coefficients $H$ are designated by $\{H_{i,j}\}$.}
 \label{fig:setting}
\end{figure}
The transmitter is nomadic \cite{fullpaper}, that is the codebook
that is used $f$, is unknown to the agents, but is fully known to
the final destination. This way the agents treat input signals not
accounting for the coded transmission, in a CEO or multiple WZ
approach. All the reported results in this paper assume that the
transmitter is limited to using only a circular-symmetric complex
Gaussian codebook. Notice that the Gaussian codebook is not
necessarily optimal, (a counter example exists for the non fading
case, where using binary signaling at the transmitter with a simple
two level demapper at the agents can outperform the Gaussian
signaling scheme, see \cite{fullpaper}). However, the Gaussian
codebook does provide a good candidate, as for
$C_i\rightarrow\infty$ and $C\rightarrow 0$ the Gaussian codebook is
indeed optimal.

In addition to the nomadism, the transmitter has no information
regarding $\boldsymbol{H}=\{H(k)\}_{k=1}^n$, where
\begin{equation*}
H(k)=\left[\begin{array}{l}H_1(k)\\ \vdots \\
H_r(k)\end{array}\right],
\end{equation*}
while the final destination is fully informed about
$\boldsymbol{H}$. By default, each agent has the full CSI
$\boldsymbol{H}$. However, many of the presented schemes require
each agent to know only its own channel coefficients
$\boldsymbol{H}_i$, as is stated in the text.

Although the transmitter is unaware of the channel realizations
$\boldsymbol{H}$, it does have the full knowledge of the channel
statistics, as well as $\{C_i\}$, which is used to calculate the
rate in which the transmitter will encode its messages.
Alternatively, higher layer control layers can indicate the
code-rate which is to be used, based on an ACK/NACK mechanism.

As said, the multiplexing gain of any scheme describes the scaling
laws of its capacity, as $P$ is increased \cite{TseZheng2003}.\\
\emph{Definition}: The multiplexing gain of a scheme is defined as
\begin{equation}
m=\lim_{P\rightarrow\infty}\frac{R}{\log(P)},
\end{equation}
whereas the diversity is defined by
\begin{equation}
d=\lim_{P\rightarrow\infty}-\frac{\log\left(\Pr\{\mathrm{outage}\}\right)}{\log(P)}.
\end{equation}

We will use $\doteq$ and $\dotle$ to denote equality and
respectively inequality, under the operation
$\lim_{P\rightarrow\infty}\frac{\log_2(\cdot)}{\log_2(P)}$. The norm
$|V|^2$ for a vector $V$ is defined as $|V|^2=\sum_i |V_i|^2$.

\section{elementary compression scheme}\label{sec:simple}
In this section, a scheme that incorporates elementary compression
at the agents is analyzed. By elementary compression, we mean
compression process that does not use the correlations between
$\{\boldsymbol{Y}_i\}$, and thus does not require the agents to have
full CSI, rather, they just need $\boldsymbol{H}_i$. In addition,
the implementation of such compression and especially the
decompression are rather simple and realized with low complexity
algorithms at the agents and the final destination.
\subsection{Ergodic Channel}
We first propose an achievable rate for general ergodic channels.
\begin{prop}\label{prop:EC_MI}\textit{
An achievable rate for an ergodic channel, with elementary
compression is
\begin{equation}
R_{EC}=I(U^r;X|H),
\end{equation}
($EC$ stands for elementary-compression) with the constraints:
\begin{equation}
I(U_i;Y_i|H)\leq C_i,\ i=1,\dots,r,
\end{equation} where
\begin{equation}
P_{X,Y^r,U^r,H}(x,y^r,u^r,h)=P_X(x)P_H(h)\prod_{i=1}^r
P_{Y_i|X,H}(y_i|x,h)P_{U_i|Y_i,H}(u_i|y_i,h).
\end{equation}
}
\end{prop}
The proof involves the random generation of codewords,
$\boldsymbol{U}_i$ according to $\prod_{k=1}^n
P_{U_i|H}(U_i(k)|H(k))$, as done in standard rate-distortion
problems with non-casual side information ($\boldsymbol{H}$). These
codebooks are made available to both all encoders and decoder. The
proof appears in Appendix \ref{app:EC_MI}.

Applying Proposition \ref{prop:EC_MI} to the Gaussian MIMO channel,
one gets
\begin{prop}\label{prop:sc_ff}\textit{An achievable rate for ergodic setting using elementary compression is equal to:
\begin{equation}\label{eq:sq_R}
R_{EC}=\max_{Q\in \mathcal{P}, \{q_i:\mathbb{C}^{[r\times
t]}\rightarrow\mathbb{R}_{+}\}_{i=1}^r}\mathrm{E}_H\left[\log_2\det\left(I_r+\mathrm{diag}\left(1-2^{-q_i(H)}\right)_{i=1}^r
HQH^*\right)\right],
\end{equation}
when the maximization in (\ref{eq:sq_R}) is such that
\begin{equation}
\label{eq:sq_con}
\mathrm{E}_H\left[\log_2\left((2^{q_i(H)}-1)\left(H_iQH_i^*+1\right)+1\right)\right]\leq
C_i,\ i=1,\dots,r,
\end{equation}
where
\begin{equation}
\mathcal{P}=\{Q: Q_{i,j}=0\ \mathrm{for}\ i\neq j,\ Q_{i,i}\geq 0,\
\mathrm{trace}(Q)\leq P\}.
\end{equation}
}
\end{prop}

Here each agent employs the elementary compression scheme which is
based on an underlying additive circular-symmetric complex Gaussian
noise channel $U_i=Y_i+D_i$, where $D_i$ is the compression noise.
As in the Gaussian CEO problem, there is a difference between the
used formulation and the backward channel $Y_i=U_i+D_i$, used for
standard rate-distortion compression.

Another issue here is that the known fading affects the variance of
the compression noise. The quantization noise is
circularly-symmetric complex Gaussian with variance $P_{D_i}(k)$
that depends on $\boldsymbol{H}(k)$. Let us further define
$q_i(h)\triangleq I(Y_i;U_i|X,H=h)$, which, due to the Gaussian
model, is equivalent to:
\begin{equation}\label{eq:q_def}
\frac{1}{1+P_{D_i}} = 1-2^{-q_i}.
\end{equation}
Notice that for all $i=1,\dots,r$, $q_i$ is a function of
$H=\{H_i\}_{i=1}^r$ and thus
is a random variable. 


It is easy to verify that the optimization problem in Proposition
\ref{prop:sc_ff} includes a concave objective function
(\ref{eq:sq_R}) but a non-convex domain (\ref{eq:sq_con}). Notice
that since $HQH^*$ is distributed the same as $HVQV^*H^*$ for any
unitary $V$, we can still limit the search for optimal $Q$ to
non-ordered elements of a diagonal $Q$ (which is the set
$\mathcal{P}$).

\begin{rem} Despite the name elementary compression, it requires
an infinite number of codebooks at the agents and the final
destination, since they should correspond to infinitely many fading
coefficients. 
\end{rem}

\subsubsection{An Achievable Rate when $r,t\rightarrow\infty$} Let us
consider the case where $r=\tau t$, symmetric agents with constant
total capacity from the agents to the final destination
($C_i=C_t/r$, $i=1,\dots,r$).
Such scheme can account for bottleneck effects in the channel
between the agents and the final destination. Let us take
$r\rightarrow\infty$, and find the limiting rate which is reliably
supported by the scheme
($\tilde{\tau}\triangleq\frac{\min\{r,t\}}{r}$). First consider that
when $t\rightarrow\infty$, we have $\frac{|H_{i_*}|^2}{t}\rightarrow
1$, almost surely. Applying this to (\ref{eq:q_def}) and
(\ref{eq:sq_con}), and also setting $P_D$ to be identical for all
the agents (the maximal $P_D$, unlike what was done in
(\ref{eq:neq_1}), which set $P_D$ to be the minimal), we get
\begin{equation}
\lim_{r\rightarrow\infty} R_{EC} \leq
\tilde{\tau}\lim_{r\rightarrow\infty} r \mathrm{E}_H
\left[\log_2\left(1+\frac{ P
\lambda/t}{1+P_{D_*}}\right)\right]
=\tilde{\tau}\lim_{r\rightarrow\infty} r
\mathrm{E}_{\nu}\left[\log_2\left(1+\frac{\tau\tilde{\tau}P
\nu}{1+\frac{P+1}{2^{C_t/r}-1}}\right)\right],
\end{equation}
where $\lambda$ is one of the unordered eigenvalues $\{\lambda_i\}$
and $\nu\triangleq\frac{\lambda}{\min\{t,r\}}$ is a random variable
with some finite mean. We can exchange the order of the expectation
and the limit due to dominant convergence
\begin{equation}\label{eq:sc_achv}
\lim_{r\rightarrow\infty}R_{EC}\leq\tilde{\tau}\mathrm{E}_{\nu}\left[\lim_{r\rightarrow\infty}r\log_2\left(1+\frac{\tau\tilde{\tau}P
\nu}{1+\frac{
P+1}{2^{C_t/r}-1}}\right)\right]=\tilde{\tau}\mathrm{E}_{\nu}\left[C_t\frac{\tau\tilde{\tau}P\nu(1+
P)}{(1+P)^2}\right]= C_t\frac{P}{1+P},
\end{equation}
where $\mathrm{E}\nu=\max\{\tau,\frac{1}{\tau}\}$. Since also
$\mathrm{argmax}_{1\leq i\leq r}\frac{|H_{i}|^2}{t}\rightarrow1$,
the inequality in equation (\ref{eq:sc_achv}) is in fact an
equality. Thus we have the following corollary:
\begin{corr}\label{corr:ec_lim} \textit{In the limit of $r,t\rightarrow\infty$, an achievable rate using elementary compression is
$C_t\frac{P}{1+P}$.}
\end{corr}
\emph{Discussion:} This result can be explained by noticing that the
MIMO channel capacity is approximately linear with $r$ when $P$ is
fixed, which leaves the fixed $C_t$ to limit the performance, where
we can not get to $C_t$ because of the nomadic setting. In addition,
this rate reaches $C_t$ in the limit of large $P$, as expected.
Notice that the rate (\ref{eq:sc_achv}) does not depend on the ratio
between the number of receive and transmit antennas, $\tau$. This is
because the signal to noise ratio (SNR) at the final destination,
from every antenna, is very small
($C_t/r=\log(1+\frac{1+P}{D})\rightarrow C_t/r=\frac{1+P}{D}$). So
that the total SNR at the final destination is
$\frac{P}{1+D}=\frac{P C_t}{C_t+(1+P)r}$, and the achievable rate
can be calculated as (small $P'$):
\begin{equation}\label{eq:sc_heur}
\log_2|I+P'/tHH^*|\rightarrow rP'.
\end{equation}
Notice that (\ref{eq:sc_heur}) indeed does not depend on $\tau$.
Taking $P'= \frac{P C_t}{C_t+(1+P)r}$ in (\ref{eq:sc_heur}) results
with (\ref{eq:sc_achv}).

\subsection{Block fading channel}
For block fading channel, the Shannon capacity is  zero, and the
concept of rate-vs-outage is the leading figure of merit.
\subsubsection{Rate vs Outage}
\begin{prop}\textit{
The rate-vs-outage region for the block fading channel, is
calculated using the same equations (\ref{eq:sq_R}) and
(\ref{eq:sq_con}), used for the fast fading channel only without the
expectation over $H$. This results with an achievable outage
probability for rate $R$, calculated as
\begin{equation}\label{eq:sq_R_outage}
\Pr(outage)=\min_{Q\in\mathcal{P}}\Pr\left(R>\log_2\det\left(I_r+\mathrm{diag}\left(\frac{2^{C_i}-1}{2^{C_i}+H_iQH_i^*}\right)_{i=1}^rHQH^*\right)\right).
\end{equation}}
\end{prop}
The underlying MIMO channel enables us to analyze the proposed
schemes using the diversity multiplexing tradeoff.
\subsubsection{Diversity Multiplexing Tradeoff
(DMT)}\label{subsec:DMT1}

An analysis for the diversity-multiplexing tradeoff is given next.
The diversity and multiplexing are defined in the end of section
\ref{sec:setting}. Since our links are lossless, any outage event in
the system is due to the underlying block fading channel. Thus, we
fix all these links to carry information in the rate
\begin{equation}\label{eq:ci_def}
C_i=\frac{m}{r}\log_2(P)+\eps,\ i=1,\dots,r
\end{equation}
where $m\leq\min\{r,t\}$ is the multiplexing gain which is used by
the system, and $\eps>0$ is some fixed positive constant.

\begin{prop}\label{prop:optDMT}\textit{ The DMT $d(m)$ of any scheme with $C_i$ as in (\ref{eq:ci_def}) and non-ergodic block fading underlying channel, is
upper bounded by the minimum between the piecewise linear function
of $(k,(r-k)(t-k))$, for $k=0,\dots,\min\{t,r\}$ and
\begin{equation}
t\left(1-\frac{m}{r}\right),
\end{equation} where $m$ stands for the multiplexing gain.}
\end{prop}

For example, the maximum diversity achieved here is with $m=0$,
which results with $d(0)=t$, which is smaller than $rt$. This result
can be understood by considering that when $m=0$, the capacity of
the links between the agents and the final destination are very
small. So that getting good channel between the transmitter and only
one agent will not suffice to forward the information. So we need a
good channel at every agent, which results with diversity order of
$t$ and not $rt$.

An implication of the result is with respect to the MIMO broadcast
channel. In order to achieve the full multiplexing gain in a MIMO
broadcast channel, the transmitter is required to have full CSI
\cite{WeingartenSteinbergShamai2006}. Here, an elementary
compression scheme, with limited cooperation between the
destinations achieves the full multiplexing gain without channel
state knowledge at the transmitter (which usually requires some
feedback). Further, such cooperation is usually easier to obtain
when the destinations are co-located.

\begin{proof}
The proof is based on the cut-set bound \cite{CoverThomas}. For any
covariance constraint $\mathrm{E}[XX^*]=Q$, and channel $H$, any
achievable rate is upper bounded by the cut-set bound, for any cut
$\mathcal{S}\subseteq\{1,\dots,r\}$
\begin{equation}
R_c=I(X;Y_\mathcal{S}|H)+\sum_{j\in\mathcal{S}^C}C_j=\log_2\det(I_{|\mathcal{S}|}+H_\mathcal{S}QH_\mathcal{S}^*)+(r-|\mathcal{S}|)\left[\frac{m}{r}\log_2(P)+\epsilon\right].
\end{equation}
So that for any scheme that achieves the rate $R(H)$ for channel
$H$, with input covariance $Q$, the probability of outage is limited
by
\begin{equation}
\forall\ R^*>0:\ \Pr\{R(H)<R^*\}\geq\Pr\{R_c(H)<R^*\}.
\end{equation}
Now we can calculate the upper bound on the DMT:
\begin{multline}\label{eq:DMT_out}
d(m)\leq-\lim_{P\rightarrow\infty}\frac{\log(\Pr(\mathrm{outage}))}{\log(P)}=\\
\min_{\mathcal{S}\subseteq\{1,\dots,r\}}-\lim_{P\rightarrow\infty}\frac{\min_{Q\in\mathcal{P}}\log\Pr\left(\log_2\det\left(I_{|\mathcal{S}|}+H_\mathcal{S}QH_\mathcal{S}^*\right)<\frac{m}{r}|\mathcal{S}|\log(P)-|\mathcal{S}^C|\eps\right)}{\log(P)}.
\end{multline}
Using \cite{TseZheng2003}, for each $\mathcal{S}$ we get that the
diversity $d_\mathcal{S}(m)$ is the piecewise linear function
connecting points $(k,(|\mathcal{S}|-k)(t-k))$, with
$\frac{|\mathcal{S}|m}{r}$ as the argument. Next, we need to
minimize this $d_\mathcal{S}(m)$ over all subsets $\mathcal{S}$.
Since $\Pr\{0<-r\eps\}=0$, we can limit the search space to subsets
that include at least one element. Define $s=|\mathcal{S}|$, so that
we can use
\begin{equation}\label{eq:DMT_b1}
\Pr\{\mathrm{outage\ with\ }s\} \dotge s
P^{-d_{\mathcal{S}}(m)}\doteq P^{-d_{\mathcal{S}}(m)}.
\end{equation}
Let us use the underlying functions of $d_\mathcal{S}(m)$, before
applying the piecewise linear operation
\begin{equation}\label{eq:no_pice}
\min_{1\leq s\leq r}
\left(s-\frac{sm}{r}\right)\left(t-\frac{sm}{r}\right).
\end{equation}
The minimum of (\ref{eq:no_pice}) is obtained by either taking $s=1$
or $s=r$, regardless of $m$. Since the piecewise linear function
exhibits the same behavior, we get Proposition \ref{prop:optDMT}.
\end{proof}
\begin{corr}\label{corr:DMT_EC}
\textit{The elementary compression achieves the full multiplexing
gain, but fails to achieve the DMT.}
\end{corr}
\emph{Proof for Corollary \ref{corr:DMT_EC}}
\begin{enumerate}
\item Next we show that elementary compression suffices to achieve the
full multiplexing gain $\bar{m}=\min\{r,t\}$. The first step is to
lower bound (\ref{eq:sq_R}) by a specific choice of $Q$ and
$P_{D_i}$. We can lower bound (\ref{eq:sq_R}) by taking
$Q=\frac{P}{t}I_t$ and the following suboptimal quantization noise
power $P_{D_i}=\frac{P/t|H_i|^2+1}{2^{C_i}-1}$ and by further taking
$\breve{P}_D\triangleq P_{D_{\breve{i}}}$,
$\breve{i}=\mathrm{argmax} \{P_{D_i}\}$
\begin{multline}\label{eq:neq_1}
R_{EC}\geq\mathrm{E}_H
\log_2\det\left(I_r+\frac{P}{t} \frac{1}{1+\breve{P}_D}\mathrm{diag}(\lambda_1,\dots,\lambda_r)\right)=\\
\mathrm{E}_H\left[\sum_{i=1}^r \log_2\left(1+\frac{P
\lambda_i/t}{1+\breve{P}_D}\right)\right]=\mathrm{E}_H\left[\sum_{i=1}^r
\log_2\left(1+\left(1+\frac{|H_{\breve{i}}|^2\frac{P}{t}+1}{P^{\frac{\bar{m}}{r}}2^\epsilon-1}\right)^{-1}\frac{P
\lambda_i}{t}\right)\right],
\end{multline}
where $\{\lambda_i\}$ are the eigenvalues of $HH^*$. Now since for
$i=1,\dots,\bar{m}$ we have that $\lambda_i>0$,
\begin{equation}
\lim_{P\rightarrow\infty}\frac{\log_2\left(1+\left(1+\frac{|H_{i^*}|^2\frac{P}{t}+1}{P^{\frac{\bar{m}}{r}}2^\epsilon-1}\right)^{-1}\frac{P
\lambda_i}{t}\right)}{\log_2(P)}=\frac{\bar{m}}{r},
\end{equation}
we get
\begin{equation}
\lim_{P\rightarrow\infty}\frac{R_{EC}}{\log_2(P)}=\bar{m}.
\end{equation} \hfill{\QED}

\item As for the DMT achieved by elementary compression ($d_{EC}$), upper bound the outage probability from equation
(\ref{eq:sq_R_outage}), and calculate the resulting diversity
\begin{multline}\label{eq:DMT_out_EC}
d_{EC}(m)\leq-\lim_{P\rightarrow\infty}\frac{\log\Pr\left(\log_2\det\left(I_{r}+P\mathrm{diag}\left\{\frac{2^{C_i}-1}{P/t|H_{i}|^2}\right\}_{i=1}^rHH^*\right)<m\log(P)\right)}{\log(P)}\\
=-\lim_{P\rightarrow\infty}\frac{\log\Pr\left(\log_2\det\left(I_{r}+t\left(P^{\frac{m}{r}}2^\epsilon-1\right)H_{\measuredangle}H^*_{\measuredangle}\right)<m\log(P)\right)}{\log(P)},
\end{multline}
where the inequality in (\ref{eq:DMT_out_EC}) is since
$\frac{2^{C_i}-1}{P/t|H_{i}|^2+2^{C_i}}\leq\frac{2^{C_i}-1}{P/t|H_{i}|^2}$
and since $\log_2\det(I+HQH^*)\leq\log_2\det(I+PHH^*)$ for any
diagonal $Q$ with $\mathrm{trace}(Q)\leq P$ \cite{TseZheng2003}. The
matrix $H_{\measuredangle}$ is defined as
\begin{equation*}
H_{\measuredangle}=\left[\begin{array}{l}\frac{H_1}{|H_1|}\\ \vdots \\
\frac{H_r}{|H_r|}\end{array}\right].
\end{equation*}
Next, since $\det(P^{\frac{m}{r}}I_r)=P^{m}$ we have the equality
\begin{multline}
\Pr\left(\log_2\det\left(I_{r}+t\left(P^{\frac{m}{r}}2^\epsilon-1\right)H_{\measuredangle}H^*_{\measuredangle}\right)<m\log(P)\right)=\\
\Pr\left(\log_2\det\left(P^{-\frac{m}{r}}I_{r}+t\left(2^\epsilon-P^{-\frac{m}{r}}\right)H_{\measuredangle}H^*_{\measuredangle}\right)<0\right).
\end{multline}
Taking the limit with respect to $P$, one gets
\begin{multline}\label{eq:DMT_out_EC3}
\lim_{P\rightarrow\infty}\Pr\left(\log_2\det\left(P^{-\frac{m}{r}}I_{r}+t\left(2^\epsilon-P^{-\frac{m}{r}}\right)H_{\measuredangle}H^*_{\measuredangle}\right)<0\right)=
\Pr\left(\log_2\det\left(
H_{\measuredangle}H^*_{\measuredangle}\right)<-(t2^\epsilon)^r\right).
\end{multline}
Using Hadamard's inequality, as long as $r>1$, the limit of the
probability in (\ref{eq:DMT_out_EC3}) is strictly larger than
zero, so that when taking the logarithm and dividing by $\log(P)$,
one gets that
\begin{equation}
d_{EC}(m)=0,
\end{equation}
for all $m>0$. So the optimal DMT is not achievable using elementary
compression, for more than a single agent $r>1$, and multiplexing
gain of more than zero. \hfill{\QED}

\item This sub-optimal DMT is since there exist correlations between the received
signals at the different agents $\frac{|H_iH_j^*|}{|H_i||H_j|}>0$,
we get that $|H_{\measuredangle}H^*_{\measuredangle}|<1$. As these
correlations decrease, for example, by taking $t$ to be very large
compared with $r$, the outage probability becomes smaller. In the
next section, we will exploit these correlations by a CEO approach,
to reach the optimal DMT.

\end{enumerate}

\section{CEO based scheme}\label{sec:WynerZiv} In this section we consider the same setting as in the
previous section, but use the technique from \cite{Confpaper}, that
is compression followed by bining, for better utilization of the
capacity of the links, exploiting the correlations between the
received signals at the agents.
\subsection{Ergodic Channel}
\begin{prop}\label{thm:WZ}\textit{
An achievable rate, for the ergodic channel, when using CEO
compression is:
\begin{equation}\label{eq:achv_MI}
R_{CEO}=\max_{P_{U_i|Y_i,H}(u_i|y_i,h)}\left\{\min_{\mathcal{S}\subseteq\{1,\dots,r\}}\left\{\sum_{i\in\mathcal{S}^C}[C_i-I(Y_i;U_i|X,H)]+I(U_{\mathcal{S}};X|H)\right\}\right\},
\end{equation}}
where
\begin{equation}
P_{X,Y^r,H,U^r}(x,y^r,h,u^r)=P_X(x)P_H(h)\prod_{i=1}^r
P_{Y_i|H,X}(y_i|h,x)P_{U_i|H,Y_i}(u_i|h,y_i).
\end{equation}
\end{prop}

\emph{Proof guidelines:} The proof involves the random generation of
$\boldsymbol{U}_i$ according to $\prod_{k=1}^n
P_{U_i|H}(u_i(k)|h(k))$, and then randomly partitioning the
resulting code book into $2^{nC_i}$ bins, as done in a WZ or a CEO
based quantization. Then, each agent selects $\boldsymbol{U}_i$
which is jointly typical with the received
$(\boldsymbol{Y}_i,\boldsymbol{H})$. It proceeds by sending the
corresponding bin index to the final destination through the
lossless link. The final destination knows $\boldsymbol{H}$ and the
bins in which $\boldsymbol{U}^r$ fall in. Finally, the destination
looks for $(\boldsymbol{X},\boldsymbol{U}^r)$ which is jointly
typical, and from deciding on the transmitted $\boldsymbol{X}$,
declares the decoded message. \\
The formal proof is by degenerating Proposition \ref{prop:achv_BC},
such that $W^r$ are constants, and the random encoding, which is
represented by $f$ is known to all parties. \hfill{\QED}


Focusing on the Gaussian channel, for the fast fading channel using
(\ref{eq:achv_MI}) the following proposition is derived.
\begin{prop}\label{prop:wz_ff}\textit{An achievable rate when using CEO compression
over Gaussian channel with fast fading is:
\begin{equation}
\label{eq:WZ_achv} R_{CEO}= \max_{\{q_i:\mathbb{C}^{[r\times
t]}\rightarrow\mathbb{R}_{+}\}_{i=1}^r}\left\{\min_{\mathcal{S}\subseteq\{1,\dots,r\}} \Bigg\{\mathrm{E}_H \Bigg[\sum_{i\in\mathcal{S}^C}[C_i-q_i(H)]+
\log_2\det\left(I_{|\mathcal{S}|}+\frac{P}{t}\mathrm{diag}\left(1-2^{-q_i(H)}\right)_{i\in\mathcal{S}}H_\mathcal{S}H_\mathcal{S}^*\right)\Bigg]\Bigg\}\right\},
\end{equation} where $H_\mathcal{S}=\{H_i\}_{i\in\mathcal{S}}$.}
\end{prop}
The Proposition is proved by using the underlying channel
$P_{U|Y,H}$ for the compression, such that the quantization noise is
independent of the signal, as done for the elementary compression
scheme in section \ref{sec:simple}. Similarly, define $P_{D_i}$ as
the power of the circular-symmetric complex Gaussian quantization
noise and $q_i(H)$ is the corresponding parameter, calculated as
(\ref{eq:q_def}).

The rate in (\ref{eq:WZ_achv}) is calculated assuming signalling
with $Q=\frac{P}{t}I_t$. The proof that such signaling indeed
maximizes the achievable rate is relegated to Appendix
\ref{app:eye_opt}. This means that the achievable rate from
Proposition \ref{prop:wz_ff} applies also to the sum-rate of multi
access channel, see \cite{TulinoLozanoVerdu2006}. Notice that
although introducing correlation in $Q$ improves the compression,
since it uses the correlation to save bandwidth, it comes on the
expense of the achievable rate, due to the reduced degrees of
freedom. Thus, the total rate is still maximized by taking
$Q=\frac{P}{t}I_t$.

\begin{rem}
The optimization over $q_i$ in the above problem is a concave
problem, and thus can be efficiently solved. The optimization
results with an achievable rate, while assuming full knowledge of
CSI ($\boldsymbol{H}$) in the final destination and in all the
agents. However, this requirement does not impose severe
limitations. This becomes
evident in the sequel where 
Correlations \ref{corr:P_big}, \ref{sub:optimal_inf_t} describe
special cases, where $q_i(H)$ is fixed, so only $\boldsymbol{H}_i$
is required at the agent.
\end{rem}

Notice that (\ref{eq:WZ_achv}) includes joint optimization over all
possible channel realizations. A simpler non-optimal approach is to
optimize separably for every channel
\begin{equation}\label{eq:WZ_achv_2}
R_{CEO,2}=\max_{Q\in\mathcal{P}}\mathrm{E}_H \Bigg[\max_{\{0\leq
q_i\}_{i=1}^r}\min_{\mathcal{S}\subseteq\{1,\dots,r\}} \Bigg\{\sum_{i\in\mathcal{S}^C}[C_i-q_i]+
\log_2\det\left(I_{|\mathcal{S}|}+\mathrm{diag}\left(1-2^{-q_i}\right)_{i\in\mathcal{S}}H_\mathcal{S}QH_\mathcal{S}^*\right)\Bigg\}\Bigg].
\end{equation}
Unlike many channel coding problems over fast fading channels
\cite{LiGoldsmith2001},\cite{BiglieriProakisShamai1998}, where there
is no loss in optimality when using different codebook for every
channel realization, here there is a strict gain to using a single
codebook, such that the decoding is done jointly over the different
realizations of $H$. So that $R_{CEO,2}<R_{CEO}$, with high
probability.

\begin{rem}\label{rem:q_i}
As in \cite{fullpaper}, $\mathrm{E}q_i$ in both (\ref{eq:WZ_achv})
and (\ref{eq:WZ_achv_2}), can be interpreted as the rate wasted on
the compression of the additive noise by the $i^{\mathrm{th}}$
agent's processing. So that, for example taking
$\mathcal{S}=\{\phi\}$ in (\ref{eq:WZ_achv}), results with the
achievable rate of $\sum_{i=1}^r [C_i-\mathrm{E}_Hq_i]$. Of course,
this represents only one of $2^{|\mathcal{S}|}$ elements within the
minimum.
\end{rem}


\begin{rem}
When the agents do not have $H$, but rather only $H_i$, that is,
each agent has only its own channel to the transmitter, and not the
channels of the other agents, the optimization in Proposition
\ref{prop:wz_ff} is done in this case over $q_i:\mathbb{C}^{[1\times
t]}\rightarrow\mathbb{R}_{+}$.
\end{rem}


Next, we give a solution to the optimization problem issued by
Proposition \ref{prop:wz_ff}, for the symmetric case with $r=2$.
Such setting results with
\begin{prop}\label{prop:explicit}\textit{An achievable rate for the symmetric setting
with $r=2$ and ergodic setting is equal to
\begin{equation}\label{eq:explic_7_prop}
R_{CEO}=2(C-\mathrm{E}_H [q_1(H,\theta)]),
\end{equation}
where ($\lceil a\rceil^+=\max\{a,0\}$)
\begin{equation}
q_i(H,\theta)=\left\{\begin{array}{cc}
\left\lceil-\log_2\left(\frac{\theta}{1+\theta}\frac{1+\frac{P}{t}|H_i|^2}{\frac{P}{t}|H_i|^2}\right)\right\rceil^+
& \theta>F_H(\frac{\Delta}{\Delta+|H_i|^2}) \\
\left\lceil-\log_2(\frac{\Delta+|H_{3-i}|^2}{\Delta}F_H(\theta))\right\rceil^+
& \theta\leq F_H(\frac{\Delta}{\Delta+|H_i|^2}),\end{array}\right.
\end{equation}
with
\begin{eqnarray}
F_H(\theta) & \triangleq & \frac{1}{2(1+\theta)}\left(1+2\theta-\sqrt{(1+2\theta)^2-4\theta(1+\theta)\frac{(\Delta+\frac{t}{P}+|H_1|^2+|H_2|^2)\Delta}{(\Delta+|H_1|^2)(\Delta+|H_2|^2)}}\right)\label{eq:explic_F}\\
\Delta & \triangleq & \frac{P}{t}\det\left(HH^*\right)\\
\end{eqnarray}
and $\theta>0$ which is set such that
\begin{equation}
\mathrm{E}_H
\log_2\det\left(I_2+\frac{P}{t}\mathrm{diag}\left(1-2^{-q_i(H,\theta)}\right)_{i=1}^2HH^*\right)=2(C-\mathrm{E}_H
[q_1(H,\theta)]).
\end{equation}
 }
\end{prop}
The proof is relegated to Appendix \ref{app:explicit}.

An intuition into the solution offered by Proposition
\ref{prop:explicit} is by considering both $q_1$ and $q_2$. For
that, assume $|H_1|^2>|H_2|^2$, then
$F_H^{-1}\left(\frac{\Delta}{\Delta+|H_1|^2}\right)<
F_H^{-1}\left(\frac{\Delta}{\Delta+|H_2|^2}\right)$ and
\begin{eqnarray}
\theta\leq F_H^{-1}\left(\frac{\Delta}{\Delta+|H_1|^2}\right):&
\left\{\begin{array}{ccc}q_1&=&-\log_2(\frac{\Delta+|H_2|^2}{\Delta}F_H(\theta))\\
q_2&=&-\log_2(\frac{\Delta+|H_1|^2}{\Delta}F_H(\theta))
\end{array}\right.\\
F_H^{-1}\left(\frac{\Delta}{\Delta+|H_1|^2}\right)<\theta<\frac{P}{t}|H_1|^2
:&
\left\{\begin{array}{ccc}q_1&=&-\log_2\left(\frac{\theta}{1+\theta}\frac{1+\frac{P}{t}|H_1|^2}{\frac{P}{t}|H_1|^2}\right)\\
q_2&=&0
\end{array}\right.\\
\frac{P}{t}|H_1|^2\leq \theta: &
\left\{\begin{array}{ccc}q_1&=&0\\
q_2&=&0.
\end{array}\right.\label{eq:zero}
\end{eqnarray}
This reveals the structure of the optimal solution, which can be
described as a variant of the famous ``water-filling". This is since
as in classic water-filling, depending on the available bandwidth,
the parameter $\theta$ determines how the compression depends on the
channel realizations. When $C$ is very large, $\theta$ is very
small, and fewer channel realizations result with $q_1=q_2=0$
(\ref{eq:zero}). When $q_1=q_2=0$ the scheme does not relay any
information regarding the channel outputs, thus saving bandwidth for
better channel realizations.

This is demonstrated in Figure \ref{fig:explcit}, for $P=7$ dB, and
$2\times2$ system, where the averaged maximum and minimum of $q_1$
and $q_2$, over 1000 channels is depicted, as function of the
Lagrangian $\theta$. It is observed that the average difference
between the two compression parameters $q_1$ and $q_2$ is about 0.4
bits/channel use. Figure \ref{fig:explcit} also draws $q_1$ and
$q_2$ for some specific channel $H$. It is seen that $q_2$ is always
larger than $q_1$, since $|H_1|^2<|H_2|^2$, for the specific
channel. Form $\theta=2$ on, $q_1=0$, which means that no
information is sent from agent $A_1$ to the final destination for
this channel realization, when $\theta\geq 2$.

\begin{figure}
\centering
 \includegraphics[width=3.5in]{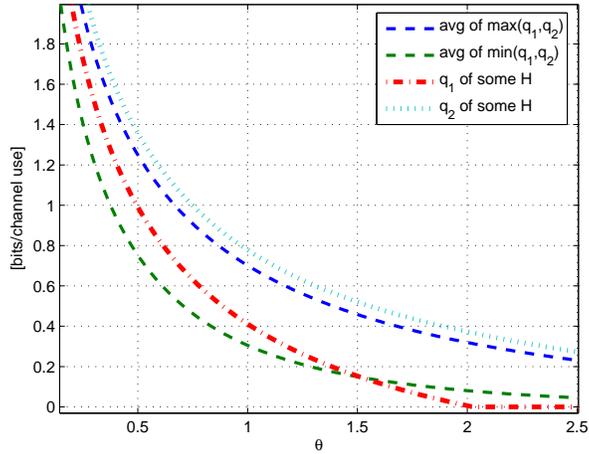}
 \caption{The resulting compression parameters $q_1$ and $q_2$, as function of $\theta$, the Lagrangian for some specific $H$, and also average results when averaging over 1000 channels $H$, $t=2$ and $P=7$
 dB.}
 \label{fig:explcit}
\end{figure}
\subsubsection{An Achievable Rate when $r,t\rightarrow\infty$}
\label{subsec:rtinfty_WZ}For the case where $r/t=\tau$, $C_i=C_t/r$
and $r\rightarrow\infty$, we repeat the suboptimal assignment and
again fix $q_i=q^*=\varepsilon C_t/r$.\\
Next, we define $q_t=rq^*=\varepsilon C_t$. Now we can write for any
$\mathcal{S}$:
\begin{multline}
\log_2\left|I_{|\mathcal{S}|}+\frac{P}{t}(1-2^{q_t/r})H_\mathcal{S}H_\mathcal{S}^*\right|=\sum_{i=1}^{|\mathcal{S}|\vee
t}\log_2(1+P/t(1-2^{q_t/r})\lambda_i)\rightarrow_{r\rightarrow\infty}\tau_\mathcal{S}
r\mathrm{E}\log_2\left(1+P(1-2^{q_t/r})\tau\tau_\mathcal{S}
\nu_\mathcal{S}\right)
\end{multline}
Where $\tau_\mathcal{S}\triangleq \frac{|\mathcal{S}|\vee t}{r}$,
$\nu_\mathcal{S}\triangleq \frac{\lambda}{|\mathcal{S}|\vee t}$ and
$\vee$ denotes min.\\
Now we can exchange the order of the expectation and the limit due
to dominant convergence:
\begin{equation}\label{eq:los_1}
\tau_\mathcal{S}\mathrm{E}\left[\lim_{r\rightarrow\infty}r\log_2\left(1+P(1-2^{q_t/r})\tau\tau_\mathcal{S}\nu_\mathcal{S}\right)\right]=\tau\tau^2_\mathcal{S}q_t
P \mathrm{E}[\nu_\mathcal{S}]=Pq_t \frac{|\mathcal{S}|}{r},
\end{equation}
since
$\mathrm{E}[\nu_\mathcal{S}]=\max\left\{\frac{|\mathcal{S}|}{t},\frac{t}{|\mathcal{S}|}\right\}$.
On the other hand, for that same $\mathcal{S}$:
\begin{equation}\label{eq:los_2}
\sum_{i\in\mathcal{S}}[C_i-q_i]=\frac{|\mathcal{S}|}{r}(C_t-q_t).
\end{equation}
Next, we set $q_t$, such that the right hand sides of
(\ref{eq:los_1}) and (\ref{eq:los_2}) are equal. This results with
the achievable rate of $R_{CEO}=C_t-q_t=C_t\frac{P}{P+1}$. Notice
that this rate is identical to the elementary compression
(\ref{eq:sc_achv}). One would expect that the Wyner-Ziv approach
will improve as $\tau$ is increased, because then the correlations
between the received signals is increased, improving the compression
rates. However, from (\ref{eq:sc_achv}), it is observed that for
small powers, the mutual information is independent of $t$, so that
also the correlations between the received signals $Y_i$ are
independent of $\tau$. In addition, from the discussion below
Correlation \ref{corr:ec_lim}, it is evident that the equivalent
signal to noise ratio when received at the final destination is very
low, so that the inter-agent correlations are also low, diminishing
the effect of the CEO compression.
\subsection{Block Fading Channel}
As in elementary compression, here we again use the rate-vs-outage
figure of merit, and then also give the DMT for the CEO based
scheme.
\subsubsection{Rate vs Outage}
For the non-ergodic block fading channel, equation
(\ref{eq:WZ_achv}), stands for the averaged mutual information.
Since the rate-vs-outage is not concave with respect to $Q$, as in
the fast fading channel, $Q=\frac{P}{t}I$ is no longer optimal
\cite{Telatar99}, and we need to optimize also over $Q$.
\begin{prop}\textit{
An achievable rate $R$ is correctly received over a block fading
channel, with an outage probability of at most $\epsilon$, as long
as the following holds (obtained from (\ref{eq:WZ_achv})):
\begin{equation}\label{eq:WZ_outage}
\Pr\left(\max_{Q\in\mathcal{P},\{0\leq q_i\leq
C_i\}_{i=1}^r}\left\{\min_{\mathcal{S}\subseteq\{1,\dots,r\}}
\left\{
\log_2\det\left(I_{|\mathcal{S}|}+\mathrm{diag}\left(1-2^{-q_i}\right)_{i\in\mathcal{S}}H_\mathcal{S}QH_\mathcal{S}^*\right)+\sum_{i\in\mathcal{S}^C}[C_i-q_i]\right\}\right\}<R\right)\leq
\epsilon
\end{equation}
where the probability is with respect to $H$.}
\end{prop}

\subsubsection{Diversity Multiplexing Tradeoff (DMT)} The CEO
approach can get to the upper bound of the DMT, and thus gives the
optimal DMT.
\begin{prop}\label{prop:DMT}\textit{The full Diversity Multiplexing
Tradeoff $d(m)$ is the minimum between the piecewise linear function
of $(k,(r-k)(t-k)$, for $k=0,\dots,\min\{t,r\}$ and
\begin{equation}
t\left(1-\frac{m}{r}\right),
\end{equation} where $0\leq m\leq \min\{r,t\}$. This tradeoff can not be achieved using the elementary compression, only using the CEO approach.}
\end{prop}
This Proposition is proved by showing that the upper bound on the
DMT from Proposition \ref{prop:optDMT} is achievable.
\begin{proof}
Consider again $C_i=\frac{m}{r}\log(P)+\epsilon$ and then fix
$q_i=0.5\epsilon$ in equation (\ref{eq:WZ_outage}). Let us write the
diversity here as $d_{CEO}$, where CEO stands for chief executive
officer
\begin{multline}\label{eq:DMT_CEO}
d_{CEO}(m)=-\lim_{P\rightarrow\infty}\frac{\log(\Pr(\mathrm{outage}))}{\log(P)}=\\
\min_{\mathcal{S}\subseteq\{1,\dots,r\}}-\lim_{P\rightarrow\infty}\frac{\min_{Q\in\mathcal{P}}\Pr\left(\log_2\det\left(I_{|\mathcal{S}|}+(1-2^{-0.5\epsilon})H_\mathcal{S}QH_\mathcal{S}^*\right)<\frac{m}{r}|\mathcal{S}|\log(P)-0.5|\mathcal{S}^C|\eps\right)}{\log(P)}.
\end{multline}
The difference between the upper bound in equation
(\ref{eq:DMT_out}) and (\ref{eq:DMT_CEO}) is with the attenuation of
$(1-2^{0.5\eps})$. Since this attenuation diminishes as $P$ gets
large, it is evident that we get the same diversity as the upper
bound.

Next, we show the achievability of the full multiplexing gain, thus
proving the DMT. We get the following achievable rate:
\begin{equation}
R_{CEO}=\bar{m}\log_2(P)+o(\log_2(P)),
\end{equation}
where $\bar{m}=\min\{r,t\}$ and
$\lim_{P\rightarrow\infty}\frac{o(\log_2(P))}{\log_2(P)}=0$. This is
since
\begin{equation}
\min_{\mathcal{S}\subseteq\{1,\dots,r\}}\Big\{|\mathcal{S}|\frac{\bar{m}}{r}\log_2(P)
+\min\{r-|\mathcal{S}|,\bar{m}\}\log_2(P)+o(\log_2(P))\Big\}=\bar{m}\log_2(P)+o(\log_2(P))
\end{equation}
is fulfilled with $\mathcal{S}=\phi$ and
$\mathcal{S}=\{1,\dots,r\}$.
\end{proof}
\subsection{An Achievable Rate For the Case of Multiple Antennas Also At the Agents}
The case of multiple antennas at the agents is different than the
above case, where only a single antenna was used by the agents, in
that now the agents can use more elaborated processing in order to
improve the overall performance. We consider here only ergodic
channel, where the block fading case follows the same line.

The channel can still be described by (\ref{eq:def}), only that now,
$Y_i(k)$ is a vector, taking values from $\mathbb{C}^{[r_i\times
1]}$, $N_i(k)\sim \mathcal{CN}(0,I_{r_i})$, and
$H_i(k)\in\mathbb{C}^{[r_i\times t]}$, again with elements that are
independently and identically distributed, according to the
circular-symmetric complex Gaussian distribution with variance of 1.

The difference between this scheme and the previous one, is that now
each agent can add non-white quantization noise (but still input
independent) to the received vector, where such dependency can
improve the resulting achievable rate, by improving the estimation
at the final destination, through better utilization of the lossless
links.
\begin{prop}\label{prop:multi_antenna}\textit{An achievable rate, over an ergodic channel, with several receiving antennas at each agent, is
\begin{multline}
\label{eq:WZ_achv_MIMO} R_{CEO}=
\max_{\{\Lambda_i(H):\mathbb{C}^{[r\times t]}\rightarrow
\mathbb{B}_i\}_{i=1}^r}\min_{\mathcal{S}\subseteq\{1,\dots,r\}}\\
\Bigg[\mathrm{E}_H
\Bigg\{\sum_{i\in\mathcal{S}^C}[C_i-\log_2|I_{m_i}+\Lambda_i^{-1}|]+\log_2\left|I_{\sum_{i\in\mathcal{S}}
m_i}+\frac{P}{t}\mathrm{diag}\left((I_{m_i}+\Lambda_i)^{-1}_{i\in\mathcal{S}}\right)H_\mathcal{S}H^*_\mathcal{S}\right|\Bigg\}\Bigg],
\end{multline}
where
\begin{equation}
H_\mathcal{S}=
\left(\begin{array}{c} \vdots \\
\Gamma_i u_i\\ \vdots
\end{array}\right)_{i\in\mathcal{S}}
\end{equation}
and
\begin{equation}
\mathbb{B}_i=\{M:\ M\in\mathbb{C}^{m'\times m'},\
m'\leq\min\{r_i,t\}, M\succeq 0\}.
\end{equation}
}
\end{prop}
To achieve this rate, each agent performs singular value
decomposition of $H_i=v_i\Gamma_iu_i$, so that
$v_i\in\mathbb{C}^{[r_i\times r_i]}$ and $u_i\in\mathbb{C}^{[t\times
t]}$ are unitary matrices, for calculating $v_i^*Y_i$. Then each
agent looks for $\boldsymbol{U}_i^n$ which is jointly typical with
$(\boldsymbol{v}_i^*)\boldsymbol{Y}_i$, when $U_i$ and $v_iY_i$ are
distributed as
\begin{equation}
U_i=v_i^*Y_i+D_i.
\end{equation}
Here $D_i$ is random vector, independent with $Y_i$, distributed as
$\mathcal{NC}(0,\Lambda_{i})$. Define $m_i=\mathrm{rank}(\Gamma_i)$
and redefine the matrix $\Gamma_i$ to include only the non-zero
elements in $\Gamma_i$. The matrix
$\Lambda_{i}\in\mathbb{C}^{[m_i\times m_i]}$ represents $m_i$ random
variables, like in the previous section, only here it is a vector
instead of a scalar.

Note that $Q=\frac{P}{t}I_t$ is optimal in (\ref{eq:WZ_achv_MIMO})
as in (\ref{eq:WZ_achv}). By assigning $r_i=1$, $\Lambda_i=P_{D_i}$
and noticing that $\Gamma_i u_i = H_i$, we see that indeed
(\ref{eq:WZ_achv_MIMO}) coincides with (\ref{eq:WZ_achv}), as
expected.

\section{Upper Bounds}\label{sec:upperbound}
In this section several upper bounds are derived, for both fast
fading and block fading cases.
\subsection{Cut-Set Upper Bound}
The simple cut-set upper bound \cite{CoverThomas}, although quite
intuitive often provides good results. This bound is very general,
and is not limited to the nomadic setting.
\begin{corr}\label{corr:cutset} \textit{Cut-set: Any achievable rate in the system is upper bounded by
the cut-set bound,
\begin{equation}
R\leq\min_{\mathcal{S}\subseteq\{1,\dots,r\}}\left[I(X;Y_{\mathcal{S}}|H)+\sum_{i\in\mathcal{S}^C}
C_i\right].
\end{equation}
For the ergodic fast fading channel, this upper bound equals
\begin{equation}
R\leq\min_{\mathcal{S}\subseteq\{1,\dots,r\}}\left[\mathrm{E}_H\log_2\det\left(I_{|\mathcal{S}|}+\frac{P}{t}H_\mathcal{S}H_\mathcal{S}^*\right)+\sum_{i\in\mathcal{S}^C}
C_i\right].
\end{equation}
Where for the block fading channel, the rate vs outage is limited by
\begin{equation}
\Pr(outage)=\min_{Q\in\mathcal{P}}\Pr\left(R>\min_{\mathcal{S}\subseteq\{1,\dots,r\}}\left[\log_2\det\left(I_{|\mathcal{S}|}+H_\mathcal{S}QH_\mathcal{S}^*\right)+\sum_{i\in\mathcal{S}^C}
C_i\right]\right).
\end{equation} }
\end{corr}
The proof is based on \cite{CoverThomas}, considering also the proof
of Proposition \ref{prop:optDMT}, and is omitted due to its
simplicity.
\subsection{Upper Bounds for Nomadic Transmitter}
The upper bounds here are calculated assuming nomadic transmitter,
who uses circular-symmetric complex Gaussian codebook. Thus they
show what cannot be achieved, no matter what processing is used at
the agents, as long as they are ignorant of the codebook used. In
the following, we first upper bound general channels, and then apply
the bound for ergodic channel and the block fading channel.
\begin{prop}\label{prop:upper_MI}\textit{
The achievable rate for reliable communication is upper bounded by:
\begin{equation}\label{eq:upper_IT}
R\leq
\min_{\mathcal{S}\subseteq\{1,\dots,r\}}\Bigg\{\sum_{i\in\mathcal{S}}
[C_i-q_i]+\frac{1}{n}
I(\boldsymbol{X};V_{\mathcal{S}^C}|\boldsymbol{H})+\frac{1}{n}\Bigg\}.
\end{equation}}
\end{prop}
\begin{proof}
We first give an information theoretic upper bound for the
achievable rate, based on \cite{fullpaper}. We define $V_i$ to be
the message sent from agent $A_i$ after receiving $n$ channel
outputs. Notice that $\boldsymbol{H}$ is fully known to all agents
and to the final destination, so they can use it to calculate the
$\{V_i\}$.

For any subset $\mathcal{S}\subseteq\{1,\dots,r\}$, the following
chain of inequalities holds:
\setlength{\arraycolsep}{0.1em}\begin{eqnarray}
\sum_{i\in\mathcal{S}} C_i
&\geq& \frac{1}{n} I(\boldsymbol{Y}^r;V_{\mathcal{S}}|V_{\mathcal{S}^C},\boldsymbol{H}) \\
&=& \frac{1}{n} I(\boldsymbol{Y}^r;V^r|\boldsymbol{H}) - \frac{1}{n} I(\boldsymbol{Y}^r;V_{\mathcal{S}^C}|\boldsymbol{H}) \\
&=& \frac{1}{n} I(\boldsymbol{Y}^r,\boldsymbol{X};V^r|\boldsymbol{H}) - \frac{1}{n} I(\boldsymbol{Y}^r,\boldsymbol{X};V_{\mathcal{S}^C}|\boldsymbol{H})\label{eq:int_x}\\
&=& \frac{1}{n} I(\boldsymbol{X};V^r|\boldsymbol{H}) -\frac{1}{n} I(\boldsymbol{X};V_{\mathcal{S}^C}|\boldsymbol{H})+
\frac{1}{n}I(\boldsymbol{Y}^r;V^r|\boldsymbol{X},\boldsymbol{H})-\frac{1}{n} I(\boldsymbol{Y}_{\mathcal{S}^C};V_{\mathcal{S}^C}|\boldsymbol{X},\boldsymbol{H})\\
&=& \frac{1}{n} I(\boldsymbol{X};V^r|\boldsymbol{H}) -\frac{1}{n} I(\boldsymbol{X};V_{\mathcal{S}^C}|\boldsymbol{H})
+\sum_{i=1}^r q_i-\sum_{i\in\mathcal{S}^C} q_i\label{eq:single_letter}\\
&=& \frac{1}{n} I(\boldsymbol{X};V^r|\boldsymbol{H}) -\frac{1}{n}
I(\boldsymbol{X};V_{\mathcal{S}^C}|\boldsymbol{H})+\sum_{i\in\mathcal{S}}
q_i.
\end{eqnarray}\setlength{\arraycolsep}{5pt}
where (\ref{eq:int_x}) is because $V_i$ is a function of
$\boldsymbol{Y}_i$ and $\boldsymbol{H}$, so we have the Markov chain
$V_i-\{Y_i,H\}-X$ and $q_i$ is defined by $q_i\triangleq
\frac{1}{n}I(\boldsymbol{Y}_i;V_i|\boldsymbol{X},\boldsymbol{H})$.
By changing order we get
\begin{equation}
\frac{1}{n} I(\boldsymbol{X};V^r|\boldsymbol{H})\leq
\sum_{i\in\mathcal{S}} [C_i-q_i] +\frac{1}{n}
I(\boldsymbol{X};V_{\mathcal{S}^C}|\boldsymbol{H}).
\end{equation}
Next we utilize Fano's inequality
\begin{eqnarray}
R&\leq& \frac{1}{n}H(M)=\frac{1}{n}I(M;V^r,F|\boldsymbol{H}) + \frac{1}{n}H(M|F,V^r,\boldsymbol{H})\\
&\leq& \frac{1}{n}I(M;V^r,F|\boldsymbol{H}) + P_e\\
&\leq& \frac{1}{n}I(M,F;V^r|\boldsymbol{H}) + P_e\\
&\leq& \frac{1}{n}I(\boldsymbol{X}(M,F);V^r|\boldsymbol{H}) + P_e\\
&\leq& \sum_{i\in\mathcal{S}} [C_i-q_i] +\frac{1}{n}
I(\boldsymbol{X};V_{\mathcal{S}^C}|\boldsymbol{H})+P_e.
\end{eqnarray}
\end{proof}
The following lemma, which is proved in the Appendix, is required
for obtaining computable upper bounds (single letter upper bound).
\begin{lem}\label{lem:EPI_1} If the transmitter is nomadic, so the agents have
no decoding ability, and the transmitter uses Gaussian codebooks,
the following inequality holds for any
$\mathcal{S}\subseteq\{1,\dots,r\}$:
\begin{equation}\label{eq:upper_1}
\frac{1}{n}I(\boldsymbol{X};V_\mathcal{S}|\boldsymbol{H}=\boldsymbol{h})\leq
m
\log_2\left(\prod_{k=1}^n\left|I_{|\mathcal{S}|}+\Lambda_{\mathcal{S}}(k)\right|^{\frac{1}{nm}}
-\prod_{k=1}^n\left|W_\mathcal{S}(k)\right|^{\frac{1}{nm}}\right)
\end{equation}
where $\Lambda_{\mathcal{S}}(k)\triangleq
H_\mathcal{S}(k)QH_\mathcal{S}^*(k)$,
\begin{equation}\label{eq:W_def}
W_\mathcal{S}(k)\triangleq \left\{\begin{array}{cc}
QH_\mathcal{S}(k)^*\mathrm{diag}\left(2^{-q_i(\boldsymbol{h})}\right)_{i\in\mathcal{S}}H_\mathcal{S}(k)
&
|\mathcal{S}|> t\\
\mathrm{diag}\left(2^{-q_i(\boldsymbol{h})}\right)_{i\in\mathcal{S}}H_\mathcal{S}(k)QH_\mathcal{S}(k)^*
& |\mathcal{S}|\leq t
\end{array}\right.
\end{equation}
$q_i(\boldsymbol{h})\triangleq
\frac{1}{n}I(Y_i^n;V_i|\boldsymbol{X},\boldsymbol{H}=\boldsymbol{h})$
and $m\triangleq \min\{t,|\mathcal{S}|\}$.
\end{lem}
Since $HQH^*$ is distributed the same as $HU^*\Sigma UH^*$, when $U$
is a unitary matrix and $\Sigma$ is diagonal, $Q$ can be restricted
to be diagonal in (\ref{eq:upper_1}). However, unlike the achievable
rate, which is a concave function of $Q$, so that $Q \propto I$ is
optimal, the right hand side of (\ref{eq:upper_1}) is not concave in
$Q$, thus in the sequel, we let $Q$ be such that $Q\in\mathcal{P}$.
Notice that the inequality in (\ref{eq:upper_1}) is tight when the
channel is $H=(1,\dots,1)^T$, which corresponds to the Gaussian CEO
problem with quadratic distortion \cite{Oohama2005}.
\subsubsection{Upper Bound for Fast Fading Channel} We begin the
derivation of an upper bound for the fast fading channel by
evaluating the bound of Lemma \ref{lem:EPI_1} for the fast fading:
\begin{corr}\label{corr:infinity_fading}
In the limit of $n\rightarrow\infty$, due to the ergodic fading
process:
\begin{equation}\label{eq:upper_infty}
\lim_{n\rightarrow\infty}
\frac{1}{n}I(\boldsymbol{X};V_\mathcal{S}|\boldsymbol{H}=\boldsymbol{h})\leq
F(\mathcal{S},q_\mathcal{S})
\end{equation}
where
\begin{equation}\label{eq:F_def}
F(\mathcal{S},q_\mathcal{S})\triangleq
m\log_2\left(2^{\frac{1}{m}\mathrm{E}_{H(1)}\log_2|I+\Lambda_\mathcal{S}|}-2^{\frac{1}{m}\mathrm{E}_{H(1)}\log_2|W_\mathcal{S}|}\right),
\end{equation}
and we use the notation $q_i\triangleq q_i(\boldsymbol{h})$ and
$q_\mathcal{S}\triangleq\{q_i\}_{i\in\mathcal{S}}$, and
$\Lambda_\mathcal{S}=\Lambda_\mathcal{S}(1)$,
$W_\mathcal{S}=W_\mathcal{S}(1)$. Consequently,
(\ref{eq:upper_infty}) can be averaged over the channels:
\begin{equation}
\lim_{n\rightarrow\infty}
\frac{1}{n}I(\boldsymbol{X};V_\mathcal{S}|\boldsymbol{H})\leq
F(\mathcal{S},q_\mathcal{S}).
\end{equation}
\end{corr}
The dependence of $F$ from (\ref{eq:F_def}) on $q_i$, stems from
the definition of $q_i$, as the bandwidth used for the noise
compression, and is essential for the bound, as it is used for
connecting the bandwidth for the signal compression to the
achievable rate. Combining proposition \ref{prop:upper_MI} with
Corollary \ref{corr:infinity_fading} above, we get the following
proposition:
\begin{prop}\label{prop:UB1}\textit{ The achievable rate of a nomadic transmitter, over fast fading channel, is upper bounded
by:
\begin{equation}\label{eq:ub_maxmin}
R\leq \max_{Q\in\mathcal{P},\{0\leq q_i\leq
C_i\}}\left\{\min_{\mathcal{S}\subseteq\{1,\dots,r\}}
\left\{F(\mathcal{S}^C,q_\mathcal{S})
+\sum_{i\in\mathcal{S}}[C_i-q_i]\right\}\right\}.
\end{equation}}
\end{prop}
\begin{rem}\label{rem:symm_fast}
When $C_i=C$ for $i=1,\dots,r$, then the argument which is maximized
over $\{q_i\}_{i=1}^r$ in (\ref{eq:ub_maxmin}), is symmetric in
$\{q_i\}_{i=1}^r$. Since the argument is also concave in
$\{q_i\}_{i=1}^r$, for $C_i=C$, equation (\ref{eq:ub_maxmin}) is
maximized by $q_i=q^*$ for $i=1,\dots,r$. So that for the symmetric
case:
\begin{equation}\label{eq:ub_maxmin_sym}
R\leq \max_{Q\in\mathcal{P},0\leq q^*\leq
C}\min_{\mathcal{S}\subseteq\{1,\dots,r\}}
\left\{F(\mathcal{S}^C,q^*) +|\mathcal{S}|[C-q^*]\right\}.
\end{equation}
\end{rem}
Following remark \ref{rem:symm_fast}, we give a special case where
the upper bound in proposition \ref{prop:UB1} is tight. Notice that
in this case, the optimal compression strategy used by the agents,
is with fixed $q^*=q_i$. This means that the each agent is required
to know only its own $H_i$, and not the other agents $\{H_j\}_{j\neq
i}$. Furthermore, notice that this conclusion is due to the tight
upper bound, and is not trivially obtained from the achievable rate
(\ref{eq:WZ_achv}) alone.
\begin{corr}\label{corr:P_big}
The CEO approach is optimal for infinite transmission power,
$Q=\frac{P}{t}I$, and $C_i=C$, $i=1,\dots,r$
\end{corr}
\label{corr:optimal_P_inf}
Here we take $P\rightarrow\infty$, and fixed $t$ and $r$.\\
\begin{proof}
We show it for $r\leq t$, where the proof for $r>t$ follows the
same lines.\\
\emph{The achievable rate:} Taking $P\rightarrow\infty$ and
optimizing over $q_{CEO}$ (where $q_i=q_{CEO},~i=1,\dots,r$ in
equation (\ref{eq:WZ_achv})) instead of over $\{q_i\}$, results
with:
\begin{equation}\label{eq:achv_infty_power3}
\frac{1}{n}I(\boldsymbol{X};V_{\mathcal{S}}|\boldsymbol{H})=m\log_2(P)+\mathrm{E}_{H}\log_2|\frac{1}{t}H_\mathcal{S}H_\mathcal{S}^*|+m\mathrm{E}_H\log_2\left(1-2^{-q_{CEO}}\right)+o(P),
\end{equation}
where $o(P)\rightarrow 0$ when $P\rightarrow\infty$.\\
\emph{The upper bound:} On the other hand, taking
$P\rightarrow\infty$ equation (\ref{eq:upper_infty}) becomes
\begin{multline}\label{eq:upper_infty_power}
F(\mathcal{S},q_\mathcal{S})=m\log_2\left(2^{\log_2(P)+\frac{1}{m}\mathrm{E}_{H}\log_2|\frac{1}{t}H_\mathcal{S}H_\mathcal{S}^*|}\left(2^{o(P)}-\prod_{i\in\mathcal{S}}2^{-\frac{q_i}{|\mathcal{S}|}}\right)\right)=\\
m\log_2(P)+\mathrm{E}_{H}\log_2|\frac{1}{t}H_\mathcal{S}H_\mathcal{S}^*|+m\log_2\left(2^{o(P)}-\prod_{i\in\mathcal{S}}2^{-\frac{q_i}{|\mathcal{S}|}}\right),
\end{multline}
Since $C_i=C$, equation (\ref{eq:ub_maxmin}) is a concave symmetric
function of $\{q_i\}$, the solution is when all $\{q_i\}$ are
identical, denoted as $q_i=q_{UB}$. So (\ref{eq:upper_infty_power})
becomes
\begin{equation}\label{eq:upper_infty_power2}
F(\mathcal{S},q_{UB})=m\log_2(P)+\mathrm{E}_{H}\log_2|\frac{1}{t}H_\mathcal{S}H_\mathcal{S}^*|+m\log_2\left(2^{o(P)}-2^{-q_{UB}}\right).
\end{equation}
which is identical, in the limit, to (\ref{eq:achv_infty_power3}).
Substituting (\ref{eq:achv_infty_power3}) in (\ref{eq:WZ_achv}) and
(\ref{eq:upper_infty_power2}) in (\ref{eq:ub_maxmin_sym}) gives the
desired equality.
\end{proof}
For $P\rightarrow\infty$ and $C_i=C$, there is no need to perform
expectation over $H$ of the rightmost element of
(\ref{eq:achv_infty_power3}), since taking $q_{CEO}=q_{UB}$ results
with the optimal rate. This means that for large $P$ and symmetric
links, the compression parameters are independent of $H$, which in
turn means that the $i$-th agent needs to know only its own $H_i$.
Notice that the channel state information (CSI, $H_i$) is still
required at $i^{th}$ agent, for the determination of the codebook of
$U$ (see \cite{fullpaper}). This is unlike the classical Gaussian
Wyner Ziv problem, which does
not benefit from side information at the encoder.\\
The upper bound of proposition \ref{prop:UB1} is not tight because
the upper bound in Lemma \ref{lem:EPI_1} was obtained using the
vector version of the entropy power inequality. This inequality is
known to be tight only for proportional correlation matrices,
which is not our case. Thus the entropy power inequality
introduces a gap that prevents the bound to be tight. This gap can
be mitigated by taking into account smaller matrices. The
following proposition improves upon proposition \ref{prop:UB1} by
optimizing also over sub-matrices of $\mathcal{S}$:
\begin{prop}\label{prop:sub1}\textit{ An achievable rate of a nomadic transmitter, which uses circular-symmetric complex Gaussian signalling with total power $P$, through
agents with bandwidths $\{C_i\}$ is upper bounded by:
\begin{equation}\label{eq:upper2_maxmin}
R_{u}\triangleq \max_{Q\in\mathcal{P} \{0\leq q_i\leq C_i\}_{i=1}^r}
\left\{ \min_{\begin{array}{lll}\cup_{j=1}^r
\mathcal{Z}_j\subseteq\{1,\dots,r\},\\ i\neq j:\ \mathcal{Z}_j\cap
\mathcal{Z}_i=\phi\ \end{array}}\left\{ \sum_{j=1}^r
F(\mathcal{Z}_j,q_{\mathcal{Z}_j})+\sum_{i\in \cap_{j=1}^r
\mathcal{Z}_j^c}[C_i-q_i]\right\}\right\}
\end{equation}
where $F(\mathcal{Z}_j,q_{\mathcal{Z}_j})$ is defined as before, in
equation (\ref{eq:upper_infty}).}
\end{prop}
The proof is very simple, considering for every group of disjoint
subsets ($\{\mathcal{Z}_j\}_{j=1}^r:\
\mathcal{Z}_j\cap\mathcal{Z}_i=\phi$ when $i\neq j$) that cover
$\cup_{j=1}^r\mathcal{Z}_j=\mathcal{S}$ we can write:
\begin{equation}\label{eq:conditional_inequality}
I(\boldsymbol{X};V_\mathcal{S}|\boldsymbol{H})\leq \sum_{j=1}^r
I(\boldsymbol{X};V_{\mathcal{Z}_j}|\boldsymbol{H}),
\end{equation}
which is due to the Markov chain $V_j-X-V_i$ when $i\neq j$, and
then using the upper bound of proposition \ref{prop:UB1} again,
for every element. Since the entropy power inequality, which is
used in proposition \ref{prop:UB1} is not tight (in general) for
the Gaussian vector case, but is tight for the Gaussian scalar
case, this upper bound
can improve upon the latter.\\
For the symmetric case, where $C_i=C$ for $i=[1,\dots,r]$, due to
the concavity of (\ref{eq:upper2_maxmin}), the maximum in
(\ref{eq:upper2_maxmin}) is achieved with $q_i=q^*,\ i=[1,\dots,r]$,
so that (\ref{eq:upper2_maxmin}) is written as:
\begin{equation}\label{eq:upper_symm_1}
R_u= \max_{Q\in\mathcal{P}, 0\leq q^*\leq C}\left\{
\min_{\begin{array}{ll}\sum_{j=1}^r jk_j\leq r,\\ k_j\geq
0\end{array}}\left\{ \sum_{j=1}^r k_j F(j,q^*)+(r-\sum_{j=1}^r
jK_j)(C-q^*)\right\}\right\}
\end{equation}
By solving the above optimization problem for $\{k_j\}_{j=1}^r$ and
then solving for $q^*$ by explicitly writing $F(j,q^*)$ we can
simplify (\ref{eq:upper_symm_1}) to
\begin{corr}\label{corr:fast_upper1}The achievable rate of nomadic transmitter in the symmetric case, $C_i=C,\
i=1,\dots,r$, is upper bounded by
\begin{equation}\label{eq:upper_symm_2}
R_{us}\triangleq rC +r\max_{Q\in\mathcal{P}}\left\{\min_{1\leq
j\leq
r}\left\{\frac{1}{j}\mathrm{E}_{H_j}\log_2|I_j+H_jQH_j^*|-\log_2\left(2^C+2^{\frac{1}{j}\mathrm{E}_{H_j}\log_2|H_jQH_j^*|}\right)\right\}\right\}
\end{equation}
where $H_j$ is the fading coefficients seen by any subset of $j$
agents (since the channel is ergodic, it does not matter which
subset).
\end{corr}

The improvement of the bound from proposition \ref{prop:sub1} over
the bound from proposition \ref{prop:UB1}, is seen in the next
corollary, where the inequality (\ref{eq:conditional_inequality}) is
in fact an equality, and a conclusive result is obtained.
\begin{corr}\label{sub:optimal_inf_t}
The CEO approach is optimal for $Q=\frac{P}{t}I$ and
$t\rightarrow\infty$ while $r$ is fix.
\end{corr}
The bound (\ref{eq:upper2_maxmin}) is tight, when
$t\rightarrow\infty$ and $Q$ is a multiplicity of the identity
matrix. This is since $HQH^*$ is proportional to the identity
matrix, each agent receives independent reception. This means $r$
parallel links that can be optimized separately. Namely, when
$t\rightarrow\infty$ while $r$ is fixed we get
\begin{equation}\label{eq:inf_t}
\lim_{t\rightarrow\infty}\frac{1}{t}HH^*=I_r.
\end{equation}
\begin{proof}\\
\emph{The achievable rate:} assigning the limit (\ref{eq:inf_t}) in
(\ref{eq:WZ_achv}), we get:
\begin{equation}\label{eq:achv_inf_t}
\lim_{t\rightarrow\infty}R(H)=\max_{\{0\leq q_i\leq C_i
\}}\left\{\min_\mathcal{S}\left\{\sum_{i\in\mathcal{S}^C}[C_i-q_i]+\sum_{i\in\mathcal{S}}\log_2(1+P(1-2^{-q_i}))\right\}\right\}.
\end{equation}
Notice that (\ref{eq:achv_inf_t}) is independent of the channel realization $H$.\\
\emph{The upper bound:} On the other hand, taking $Q=\frac{P}{t}I_t$
and $\frac{1}{t}H_iQH^*_i=1$ for the calculation of $F(\{i\},q_i)$
in (\ref{eq:upper_infty}) gives $\log_2(1+P(1-2^{-q_i}))$. Assigning
back to equation (\ref{eq:upper2_maxmin}), with
$\mathcal{Z}_i=\{i\}$ results with:
\begin{equation}\label{eq:infty_max_min}
\lim_{t\rightarrow\infty}R_u=\max_{\{0\leq q_i\leq C_i\}
}\left\{\min_{\mathcal{S}}\left\{\sum_{i\in\mathcal{S}^C}[C_i-q_i]+\sum_{i\in\mathcal{S}}\log_2(1+P(1-2^{-q_i}))\right\}\right\},
\end{equation}
which equals (\ref{eq:achv_inf_t}) and proves the optimality.
\end{proof}

\subsubsection{Upper Bound for Block Fading Channels} In this
subsection, we will consider the case of $H$ distributed
independently, but once per block, such that $\boldsymbol{H}=H$. The
resulting rate in equation (\ref{eq:WZ_achv}) is actually the
average rate, supported by the scheme. In the sequel of this
subsection, we will upper bound the rate-vs.-outage of the
scheme.\\
For the upper bound, we again use:
\begin{equation}\label{eq:upper_out_basic}
R(\boldsymbol{H}=\boldsymbol{h})\leq
\max_{\{q_i\}_1^r}\min_{\mathcal{S}}\left\{\frac{1}{n}I(V_{\mathcal{S}};\boldsymbol{X}|\boldsymbol{H}=\boldsymbol{h})+\sum_{i\in\mathcal{S}^C}[C_i-q_i]\right\}.
\end{equation}
For
$I(V_{\mathcal{S}};\boldsymbol{X}|\boldsymbol{H}=\boldsymbol{h})$,
we use the upper bound of equation (\ref{eq:upper_1}). Since
$\boldsymbol{H}=H$, we get:
\begin{equation}
G(\mathcal{S},q_\mathcal{S})\triangleq m
\log_2\left(\left|I_{|\mathcal{S}|}+\Lambda_{\mathcal{S}}\right|^{\frac{1}{m}}
-\left|W_\mathcal{S}\right|^{\frac{1}{m}}\right)
\end{equation}
\begin{equation}\label{eq:outage_1}
\frac{1}{n}I(\boldsymbol{X};V_\mathcal{S}|H=\boldsymbol{h})\leq
G(\mathcal{S},q_\mathcal{S})
\end{equation}
where $\Lambda_{\mathcal{S}}= H_\mathcal{S}QH_\mathcal{S}^*$, as
before and $W_\mathcal{S}$ is defined by $W_\mathcal{S}(1)$ from
equation (\ref{eq:W_def}). Combining (\ref{eq:upper_out_basic})
and (\ref{eq:outage_1}) and noticing that $H$ is a random
variable, we get the following upper bound on the outage
$\epsilon$ vs. rate $R$:
\begin{prop}\label{prop:block_fading}\textit{An
upper bound on the achievable rate $R$, for given outage
probability $\epsilon$ is the minimal $R$ which fulfills:
\begin{equation}\label{eq:outage_2}
P\left(\max_{Q\in\mathcal{P},\{0\leq q_i\leq
C_i\}}\left\{\min_{\mathcal{S}\subseteq\{1,\dots,r\}}
\left\{G(\mathcal{S},q_\mathcal{S})
+\sum_{i\in\mathcal{S}^C}[C_i-q_i]\right\}\right\}<R\right)\leq
\epsilon.
\end{equation}}
\end{prop}
Actually, we can improve upon (\ref{eq:outage_2}), the same it was
done in proposition \ref{prop:sub1}:
\begin{equation}\label{eq:outage_3}
P\left(\max_{Q\in\mathcal{P},\{0\leq q_i\leq C_i\}
}\left\{\min_{\begin{array}{lll}\cup_{j=1}^r
\mathcal{Z}_j\subseteq\{1,\dots,r\},\\ i\neq j:\ \mathcal{Z}_j\cap
\mathcal{Z}_i=\phi\ \end{array}} \left\{\sum_{j=1}^r
G(\mathcal{Z}_j,q_{\mathcal{Z}_j}) +\sum_{i\in \cap_{j=1}^r
\mathcal{Z}_j^c}[C_i-q_i]\right\}\right\}<R\right)\leq \epsilon,
\end{equation}
but since the problem is not symmetric (due to the non-ergodic $H$),
we can not further simplify it, as in Corollary
\ref{corr:fast_upper1}. However, the limiting behavior of
(\ref{eq:inf_t}) is true also for the block fading case. Thus the
optimality of the CEO approach when $t\rightarrow\infty$ from
correlation \ref{sub:optimal_inf_t} is assured for the block fading
case as well.
\subsection{Discussion}
When considering the upper bound, several clarifications are in
order. It is known \cite{Oohama2005},\cite{fullpaper} that when no
fading is present, and the transmitter has only a single antenna,
the upper bound is in fact tight. It means that when the sum
$\sum_{i=j}^r Y_j$ is sufficient statistics, the capacity is
established. This situation changes when considering fading
channels. It is evident from \cite{KrithivasanPradhan2007}, that
when $Y_1-Y_2$ is sufficient statistics, using our technique, which
is based on the Berger-Tung CEO, is strictly sub-optimal and lattice
approach can outperform the random binning. Therefore, it is not
expected that ultimate performance is achieved, although the upper
bound proximity to the achievable rate.
\section{Numerical Example}
The achievable rates and the upper bounds for both fast fading and
block fading channels, were calculated for a $2\times 2$ system,
with $C_1=C_2=2$, for several signal to noise ratios ($P$ in dB),
and the results are presented in figure \ref{fig:performance2}. For
the fast fading, both achievable rate and upper bound are obtained
by averaging over 30 blocks, each containing 50 channel realizations
(the expectation expressed by $\mathrm{E}_{H}$ in
(\ref{eq:upper_symm_2}) and (\ref{eq:WZ_achv})). It is seen there
that the upper bound is convex, and that it is close to the
achievable rate, when using CEO compression. For the lower and
higher $P$ the bound is
tighter.\\
For block fading channel, the upper bound from (\ref{eq:outage_3})
is depicted along with the achievable rate (\ref{eq:WZ_outage}), for
outage probability of $\epsilon=10^{-2}$. The probability was
calculated using Monte Carlo simulations over 10000 different
realizations of $H$. It is seen there that the bound is again very
tight for the low SNR region, and the gap becomes higher, with
larger SNR, although it remains rather small, no more than 1 dB
throughout the figure.
\begin{figure}
\centering
 \includegraphics[width=3.5in]{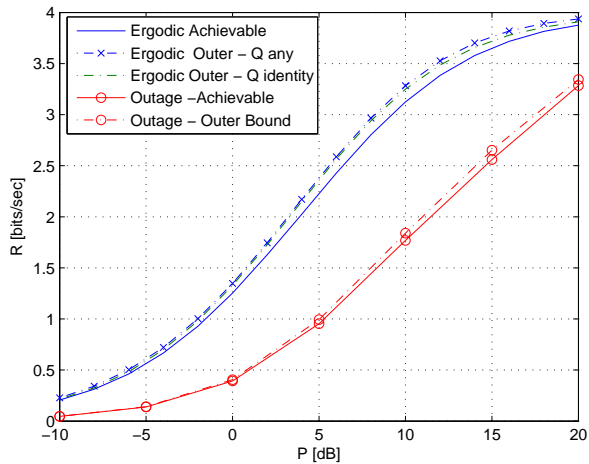}
 \caption{The achievable rates compared to the upper bounds over a $2\times 2$ system with $C=2$: for fast fading Rayleigh channel with upper bound according to an arbitrary $Q$ (Q singular) and to a fix $Q=\frac{P}{t}I_t$ (Q identity), and for block fading Rayleigh channel, with outage probability of $10^{-2}$, where the upper bound was calculated from
(\ref{eq:outage_3}). All as a function of $P$ in {dB}, where the
outage probability and the average over $H$ were done by Monte
Carlo simulations over $H$.}
 \label{fig:performance2}
\end{figure}

\section{Agents with Code Knowledge, and Fully Informed
Transmitter}\label{sec:cogniz} In this section we consider the same
model, as in the previous sections, with two differences. One
difference is that we drop the nomadity assumption, and let the
agents be able to decode messages. The second difference is that we
assume full CSI ($\boldsymbol{H}$) at the transmitter, in a non
casual sense, so that the transmitter and the agents have the same
channel state information.

We get to the following proposition, which is proved in the
Appendix.
\begin{prop}\label{prop:achv_BC}\textit{In the ergodic regime, when
the transmitter has full CSI, and the agents are cognizant of the
codebook used, the rate (\ref{eq:cong_IT_achv}) is achievable
\begin{equation}\label{eq:cong_IT_achv}
R_{cog}
=\max_{\pi}\min_\mathcal{S}\left\{\sum_{i\in\mathcal{S}}[C_i-I(U_i;Y_i|X,W^r,H)]+I(U_{\mathcal{S}^C};X|W^r,H)
+\sum_{i\in\mathcal{S}^C}
[I(W_i;Y_i|H)-I(W_i;W_{\tilde{\mathcal{T}}(\pi,i)}|H)]\right\},
\end{equation}
where $\pi$ is a permutation of $[1,\dots,r]$,
\begin{equation}
\tilde{\mathcal{T}}(\pi,i)\triangleq \{\pi_1,\dots,i\},
\end{equation}
and
\begin{equation}
P_{W^r,X,Y^r,U^r|H}(w^r,x,y^r,u^r|h)=P_{W^r|H}(w^r|h)P_{X|W^r,H}(x|w^r,h)\prod_{i=1}^r
[P_{Y_i|X,H}(y_i|x,h)P_{U_i|Y_i,W_i,H}(u_i|y_i,w_i,h)].
\end{equation}
 }
\end{prop}

The transmitter sends messages to the agents via the broadcast
channel \cite{WeingartenSteinbergShamai2006}, by using the dirty
paper coding (DPC) technique \cite{CaireShamai2003}. On top, the
transmitter also sends information to be decoded only at the final
destination, invoking the nomadic techniques of the previous scheme.
We actually extended the results of \cite{fullpaper}, to include
also DPC and a random ergodic channel. In \cite{fullpaper} Corollary
4, the superposition coding combined with the CEO technique, was
used for that setting, when no fading was present, and when the
channel was degraded. The main difference between superposition
coding and DPC is in that superposition coding lets the destined
terminal cancel the interfering transmissions (which are destined to
terminals with weaker channels) and the DPC performs precoding, so
that interference transmissions are canceled at the transmitter
(thus the name dirty paper coding).

Next, for the fading Gaussian channel, the combined final
destination decoding and DPC results with the rate
\begin{multline}\label{eq:DPC_1}
R_{DPC,1} =
\max_{Q,\pi,\{B_i,q_i\}_{i=1}^r}\min_\mathcal{S}\mathrm{E}_H\Bigg\{\sum_{i\in\mathcal{S}}[C_i-q_i]+\log_2\left|I_{|\mathcal{S}^C|}+\mathrm{diag}(1-2^{-q_i})H_{\mathcal{S}^C}\left(Q-\sum_{i=1}^rB_i\right)H_{\mathcal{S}^C}^*\right|\\
+\sum_{i\in\mathcal{S}^C}
\log_2\left(\frac{1+H_i\left(Q-\sum_{j\in\tilde{\mathcal{T}}(\pi,i)}B_j\right)H_i^*}{1+H_i\left(Q-\sum_{j\in\tilde{\mathcal{T}}(\pi,i)}B_j-B_i\right)H_i^*}\right)\Bigg\},
\end{multline}
where the maximization is over
\begin{eqnarray}
q_i&:& \mathbb{C}^{[r\times
t]}\rightarrow \mathbb{R}_+,\\
Q,B_i&:&\mathbb{C}^{[r\times t]}\rightarrow \mathbb{C}^{[t\times
t]},
\end{eqnarray}
such that $Q,B_i\succeq 0$, $Q-\sum_{1}^r B_i\succeq 0$ and
$\mathrm{E}_H[\mathrm{trace}(Q)]\leq P$. The rate in
(\ref{eq:DPC_1}) can be increased by convex hull
\cite{CaireShamai2003}, since in general, this problem is non
concave.

This rate is achieved by using $W_i$, as in \cite{CaireShamai2003},
and then $P_{U_i|Y_i,W_i,H}=P_{U_i|Y_i,H}$ remains the same as in
Proposition \ref{prop:EC_MI}. The situation in the compression
stage, is similar to when using Wyner-Ziv source compression with
decoder side information over Gaussian sources, where supplying the
side information ($W_i$) to the encoder does not improve the rate
distortion.

Although calculating (\ref{eq:DPC_1}) is hard, due to the
non-convexity of the problem, note that a sub-optimal rate can be
calculated for the symmetric case ($C_i=C$), by using the DPC such
that the maximal sum-rate is obtained, and so that
$Q=I_t\frac{P'}{t}-\sum_{j=1}^r B_j\succeq 0$, and letting
$\mathrm{E}_H[P']\leq P$. Since the problem is symmetric and the
channel ergodic, each agent decodes the same rate. The DPC
sum-rate can be obtained by the dual multi-access (concave) MIMO
channel \cite{VishwanathJindalGoldsmith2003}.
\section{Conclusion}\label{sec:conclusion}
In this paper we showed the effectiveness of several compression
techniques for decentralized reception in fast fading and block
fading MIMO channels. We proved that in many cases, the elementary
compression is sufficient to get the full-multiplexing gain. In
addition, we showed the advantages of the CEO approach, which were
evident in an asymptotic analysis and in a finite example. We
presented upper-bounds for both fast fading channel and block fading
channel, which are based on the nomadic characteristic of the
scheme, along with the EPI, and which turned out to be quite tight
even for relatively small $2\times 2$ scheme. Achievable rate for a
non-nomadic scheme was finally derived, combining the decentralized
processing with the DPC.
\section*{Acknowledgment}
This research was supported by the EU 6th framework program via
the NEWCOM network of excellence.

\appendices
\section{Useful definitions and Lemmas}
Let $P_{A_1,A_2,\dots,A_L}(a_1,a_2,\dots,a_L)$ be the probability
function of the random variables $A_1,\dots,A_L$ which take values
in $\mathcal{A}_1,\dots,\mathcal{A}_L$, respectively.\\
\emph{Definitions:}
\begin{enumerate}
\item The marginal probabilities are then defined as
\begin{equation}
P_{A_l}(a_l)=\sum_{a_{\mathcal{L}\setminus
l}\in\mathcal{A}_{\mathcal{L}\setminus l}}
P_{A_1,A_2,\dots,A_L}(a_1,a_2,\dots,a_L)
\end{equation}
($\mathcal{L}$ is the set $\{1,\dots,L\}$).
\item The conditional
probabilities are defined as:
\begin{equation}
P_{A_l|A_\mathcal{S}}(a_l|a_\mathcal{S})=\frac{P_{A_l,A_\mathcal{S}}(a_l,a_\mathcal{S})}{P_{A_\mathcal{S}}(a_\mathcal{S})},
\end{equation}
for some $\mathcal{S}\subseteq\mathcal{L}$ and $l\notin \mathcal{S}$
and $P_{A_\mathcal{S}}(a_\mathcal{S})\neq 0$.
\item As commonly done (see \cite{CoverThomas}, section 13, problem 10),
define the $\epsilon$-typical (strongly conditional typical) set
$\mathbf{T}_\epsilon$ of $\boldsymbol{a}_\mathcal{L}$ as the set for
which
$N(a_\mathcal{S},h|\boldsymbol{a}_\mathcal{S},\boldsymbol{h})=0$ for
any $a_\mathcal{S} \in
\mathcal{A}_\mathcal{S},\boldsymbol{h}\in\mathcal{H}$ such that
$P_{A_\mathcal{S}|H}(a_\mathcal{S}|h)=0$, and also
\begin{multline}\label{eq:typical_set_def}
\mathbf{T}_\epsilon(\boldsymbol{h})\triangleq\
\Big\{\boldsymbol{a}_\mathcal{L}:\ \forall
\mathcal{S}\subseteq\mathcal{L},\ \forall a_\mathcal{S} \in
\mathcal{A}_\mathcal{S},h\in\mathcal{H}\ 
\frac{1}{n}\left|N(a_\mathcal{S},h|\boldsymbol{a}_\mathcal{S},\boldsymbol{h})-P_{A_\mathcal{S}|H}(a_\mathcal{S}|h)N(h|\boldsymbol{h})\right|<
\frac{\epsilon}{ |\mathcal{A}_\mathcal{S}|} \Big\},
\end{multline}
where $N(a_\mathcal{S}|\boldsymbol{a}_\mathcal{S})$ denotes the
counting operator of the number of occurrences of the symbol
$a_\mathcal{S}$ in the vector $\boldsymbol{a}_\mathcal{S}$.
\end{enumerate}
\emph{Lemmas:}
\begin{lem}\label{lem:AEP}
For any $\epsilon>0$, there exist $n^*$ such that for all $n>n^*$
and randomly generated $\boldsymbol{a}_\mathcal{L}$ according to
$\prod P_{A_\mathcal{L}|H}(a_{\mathcal{L}}(k)|h(k))$
\begin{equation}
\Pr\{\boldsymbol{a}_\mathcal{L} \in
\mathbf{T}_\epsilon(\boldsymbol{h})\}\geq 1-\epsilon.
\end{equation}
\end{lem}
\begin{lem}\label{lem:jointly_typical}
Fix some $\mathcal{S}\subseteq\mathcal{L}$ and probability
\begin{equation}\label{eq:joint_probability}
P_{A_{\mathcal{L}},W_\mathcal{L}|H}(a_{\mathcal{L}},w_\mathcal{L}|h).
\end{equation}
Define the jointly $\epsilon$-typical set
$\mathbf{T}_\epsilon(\boldsymbol{h})$, as before,
by the joint probability (\ref{eq:joint_probability}).\\
Let $\boldsymbol{a}_\mathcal{L}^n$ be generated according to
\begin{equation}\label{eq:true_probability}
\boldsymbol{a}_\mathcal{L}\sim \prod_{k=1}^n\left\{
P_{A_{\mathcal{S}^C}|W_{\mathcal{L}},H}(a_{\mathcal{S}^C}(k)|w_{\mathcal{L}}(k),h(k))\prod_{l\in\mathcal{S}}P_{A_l|W_l,H}(a_{l}(k)|w_{l}(k),h(k))\right\},
\end{equation} where the conditional and marginal probabilities are
calculated from (\ref{eq:joint_probability}) and
$\boldsymbol{w}_{\mathcal{L}}$ is a given vector which was randomly
generated and that belongs to the set
$\mathbf{T}_\epsilon(\boldsymbol{h})$, as defined by
(\ref{eq:typical_set_def}) (that is, there exist
$\boldsymbol{a}_\mathcal{L}$
that are jointly typical with $\boldsymbol{w}_\mathcal{L}$).\\
Then the probability of the vector $\boldsymbol{a}_\mathcal{L}$
which is distributed according to (\ref{eq:true_probability}) to be
in $\mathbf{T}_\epsilon(\boldsymbol{h})$, which is defined according
to (\ref{eq:joint_probability}) is bounded by:
\begin{multline}
\Pr\{(\boldsymbol{a}_{1,\dots,L},\boldsymbol{w}_\mathcal{L})\in\mathbf{T}_\epsilon(\boldsymbol{h})\}\geq 
2^{-n[H(A_{\mathcal{S}^C}|W_{\mathcal{L}},H)-H(A_\mathcal{L}|W_{\mathcal{L}},H)+\sum_{l\in\mathcal{S}}H(A_l|W_l,H)+\epsilon_1]}
\end{multline}
\begin{multline}
\Pr\{(\boldsymbol{a}_{1,\dots,L},\boldsymbol{w}_\mathcal{L})\in\mathbf{T}_\epsilon(\boldsymbol{h})\} \leq 
2^{-n[H(A_{\mathcal{S}^C}|W_{\mathcal{L}},H)-H(A_\mathcal{L}|W_{\mathcal{L}},H)+\sum_{l\in\mathcal{S}}H(A_l|W_l,H)-\epsilon_1]}
\end{multline}
where $\epsilon_1\rightarrow 0$ as $\epsilon\rightarrow 0$.\\
\end{lem}
\begin{lem}\label{lem:gen_Markov_lem}Generalized Markov Lemma\\
Let
\begin{equation}
P_{A_{\mathcal{S}},W_\mathcal{S},Y_\mathcal{S}|H}(a_{\mathcal{S}},w_\mathcal{S},y_\mathcal{S}|h)=P_{W_\mathcal{S},Y_\mathcal{S}|H}(w_\mathcal{S},y_\mathcal{S}|h)\prod_{l\in\mathcal{S}}P_{A_t|W_t,Y_t,H}(a_t|w_t,y_t,h).
\end{equation}
Given randomly generated
$\boldsymbol{w}_\mathcal{S}\boldsymbol{y}_\mathcal{S}$ according to
$P_{W_\mathcal{S},Y_\mathcal{S}|H}$, for every $i\in\mathcal{S}$,
randomly and independently generate $N_i\geq 2^{nI(A_i;Y_i|W_i,H)}$
vectors $\boldsymbol{\tilde{a}}_i$ according to
$\prod_{k=1}^nP_{A_i|W_i,H}(\tilde{a}_{i}(k)|w_i(k),h(k))$, and
index them by $\boldsymbol{\tilde{a}}_i^{(t)}$ ($1\leq t \leq N_i$).
Then there exist $|\mathcal{S}|$ functions
$t^*_i=\phi_i(\boldsymbol{y}_i,\boldsymbol{w}_i,\boldsymbol{\tilde{a}}_i^{(1)},\dots,\boldsymbol{\tilde{a}}_i^{(N_t)})$
taking values in $[1\dots N_t]$, such that for sufficiently large
$n$,
\begin{multline}
\Pr((\{\boldsymbol{a}_{i}^{(t^*_i)}\}_{i\in\mathcal{S}},\boldsymbol{w}_\mathcal{S},\boldsymbol{y}_\mathcal{S})\in\mathbf{T}_\epsilon(\boldsymbol{h}))
\geq 1-\epsilon.
\end{multline}
\end{lem}

\begin{proof}
See \cite{CoverThomas} and \cite{Csiszar_Korner} for the proofs of
Lemmas \ref{lem:AEP}-\ref{lem:jointly_typical}, while Lemma
\ref{lem:gen_Markov_lem} is a simple extension of Lemma 3.4
(Generalized Markov Lemma) in
\cite{HanKobayashi1980}.\end{proof}

In the following, we use only $\epsilon$ and remove the distinction
between $\epsilon$ and $\epsilon_1$, for the sake of brevity.

\section{Proof of Proposition \ref{prop:EC_MI}}\label{app:EC_MI}
\subsection{Code construction:}
Fix $\delta>0$.
\begin{enumerate}
\item For the transmitter, for any codebook used, $f$
\begin{itemize}
\item Randomly choose $2^{nR_{CEO}}$ vectors $\boldsymbol{x}$, with probability
$P_{\boldsymbol{X}}(\boldsymbol{x})=\prod_k P_{X}(x(k))$.
\item Index these vectors by $M_{CEO}$ where $M_{CEO}\in[1,2^{nR_{EC}}]$.
\end{itemize}
\item For the compressor at the agents\\
For every channel realization $\boldsymbol{h}$
\begin{itemize}
\item Randomly generate $2^{nC_i}$ vectors $\boldsymbol{u}_i$ of length $n$\\ according to
$\prod_k
P_{U_i|H}(u_{i}(k)|h(k))$. 
\item Index all the generated $\boldsymbol{u}_i$ with $z_i\in[1,2^{nC_i}]$.
\end{itemize}
\end{enumerate}
\subsection{Encoding:}
Let $M$ be the message to be sent, and $f$ is the codebook used. The
transmitter then sends $\boldsymbol{x}(M,f)$ to the channel.
\subsection{Processing at the
agents:}
The $i^{th}$ agent chooses any of the $z_i$ such that
\begin{equation}\label{eq:relay_z_typical}
\big(\boldsymbol{u}_i(z_i,\boldsymbol{h}),\boldsymbol{y}_i\big)\in\mathbf{T}_{\epsilon}^{EC,i}(\boldsymbol{h}),
\end{equation}
where $\mathbf{T}_{\epsilon}^{EC,i}(\boldsymbol{h})$ is defined in
the standard way, as (\ref{eq:typical_set_def}). The event where no
such $z_i$ is found is defined as the error event $E_1$.\\
After deciding on $z_i$ the agent forwards it to the final
destination through the lossless link.
\subsection{Decoding (at the destination):}
The destination retrieves $z^r$ from the lossless links, and uses $\boldsymbol{h}$ and the random encoding $f$.\\
The destination then finds $\hat{M}$ such that
\begin{equation}
\big(\boldsymbol{x}(\hat{M},f),\boldsymbol{u}^r(\hat{z}^r)\big)\in
\mathbf{T}_{\epsilon}^{EC,3}(\boldsymbol{h}).
\end{equation}
Where $\mathbf{T}_{\epsilon}^{EC,3}$ is defined in the standard way,
as (\ref{eq:typical_set_def}). If there is no such $\hat{M}$, or if
if there is more than one, the destination chooses one arbitrarily.
Define error $E_2$ as the event where $\hat{M}\neq
M_{CEO}$.\\
Correct decoding means that the destination decides $\hat{M}=M$.
An achievable rate $R$ was defined as when the final destination
receives the transmitted message with an error probability which
is made arbitrarily small for sufficiently large block length $n$.
\subsection{Error analysis}\label{subsec:err_ana}
The error probability is upper bounded by:
\begin{equation}
\Pr\{\mathrm{error}\}=\Pr\left(\cup_{i=1}^{2} E_i\right)\leq
\sum_{i=1}^{2}\Pr(E_i).
\end{equation}
Where:
\begin{enumerate}
\item $E_1$: No $\boldsymbol{u}_i(z_i,\boldsymbol{h})$ is jointly typical with
$\boldsymbol{y}_i$.
\item $E_2$: Decoding error $\boldsymbol{x}(\hat{M},f)\neq \boldsymbol{x}(M,f)$, so that $\hat{M}\neq M$.
\end{enumerate}
Next, we will upper bound the probabilities of the individual error
events by arbitrarily small $\epsilon$.
\subsubsection{$E_1$}
According to Lemma \ref{lem:gen_Markov_lem}, the probability
$\Pr\{E_1\}$ can be made as small as desired, for $n$ sufficiently
large, as long as
\begin{equation}
C_i>I(U_i;Y_i|H).
\end{equation}
\subsubsection{$E_2$}
Consider the case where $\hat{M}\neq M$. There are $2^{nR_{CEO}}$
such vectors, and the probability of
$(\boldsymbol{x}(\hat{M},f),\boldsymbol{u}^r(z^r))$ to be jointly
typical is upper bounded by (Lemma \ref{lem:jointly_typical})
$2^{-n[I(X;U^r|H)-\epsilon]}$. Thus the rate $R_{CEO}$ is achievable
if:
\begin{equation}
R_{CEO}<I(X;U^r|H)-\epsilon,
\end{equation}
which proves Proposition \ref{prop:EC_MI}. \hfill{\QED}

\section{Proof of Optimality of $Q=\frac{P}{t}I_t$ for the Ergodic
Channel.}\label{app:eye_opt} First consider that since the channel
is unknown to the transmitter, and $VH$ is distributed as $H$ when
$V$ is unitary (eigenvectors of a non diagonal $Q$) all through
this work, $Q$ can be limited to be diagonal.

Next, for any given $q_i(H)$ and $\mathcal{S}$, we have that
\begin{equation}
\mathrm{E}_H\left[\log_2\det\left(I_{|\mathcal{S}|}+\mathrm{diag}\left(1-2^{-q_i(H)}\right)_{i\in\mathcal{S}}H_\mathcal{S}QH_\mathcal{S}^*\right)\right]
\end{equation}
is a concave function of $Q$, which is thus maximized by
$Q=\frac{P}{t}I_t$ \cite{Telatar99}. Thus it also maximizes the
maximum over all $q_i(H)$ and $\mathcal{S}$ concluding the proof.
\hfill{\QED}

Notice that this proof does not extend to (\ref{eq:sq_R}) and to
(\ref{eq:WZ_achv_2}), so that there, the optimal $Q$ may not be
proportional to identity, but is still diagonal, though.
\section{Proof For Proposition
\ref{prop:explicit}}\label{app:explicit}
In this Appendix, we give a closed solution to the $r=2$, symmetric
case. We extend what was done in \cite{fullpaper} to the ergodic
channel case, with $t>1$. Equation (\ref{eq:WZ_achv}) for the
symmetric case can be written as:
\begin{equation}\label{eq:exp1}
R_{CEO}= \max_{0\leq q^*\leq
C}\left\{\min_{\mathcal{S}\subseteq\{1,\dots,r\}}
\{|\mathcal{S}^C|[C-q^*]+F_\mathcal{S}(q^*)\}\right\},
\end{equation}
where
\begin{equation}
F_\mathcal{S}(q^*) = \max_{\{q_i:\mathbb{C}^{[r\times
t]}\rightarrow\mathbb{R}_{+}\}_{i=1}^r} \mathrm{E}_H
\log_2\det\left(I_{|\mathcal{S}|}+\frac{P}{t}\mathrm{diag}\left(1-2^{-q_i(H)}\right)_{i\in\mathcal{S}}H_\mathcal{S}H_\mathcal{S}^*\right)
\end{equation}
such that
\begin{equation}
\mathrm{E}_H[q_i(H)]=q^*,\ i=1,\dots,r.
\end{equation}
Since the channel is ergodic, and the scheme symmetric, the users
will be equivalent, and due to the concavity of the problem, the
optimal solution is characterized by $q^*=\mathrm{E}_H [r_i(H)]$.
That is, equal bandwidth that is wasted by all users on the noise
quantization. By writing the equation this way, the ergodic nature
of the channel is used, such that the channel randomness is limited
to within $F_\mathcal{S}$. Since $F_\mathcal{S}$ is an increasing
function of $q^*$, when solving it, the solution of
(\ref{eq:WZ_achv}) is readily available numerically. So we are left
with the concave problem of finding $F_\mathcal{S}$.

Since $F_{\{1,\dots,r\}}(q^*)$ is an increasing function of $q^*$,
and $r(C-q^*)$ is a decreasing function of $q^*$, the point
$F_{\{1,\dots,r\}}(q^*)=r(C-q^*)$ exists, and further, it is an
upper bound to the achievable rate. Next, using Hadamard inequality
we have that for any $\mathcal{S}$
\begin{multline}
\log_2\det\left(I_2+\frac{P}{t}\mathrm{diag}(1-2^{-q_i})_{i=1}^rHH^*\right)\leq\log_2\det\left(I_\mathcal{S}+\frac{P}{t}\mathrm{diag}(1-2^{-q_i})_{i\in\mathcal{S}}H_\mathcal{S}H_\mathcal{S}^*\right)\\
+\log_2\det\left(I_{\mathcal{S}^C}+\frac{P}{t}\mathrm{diag}(1-2^{-q_i})_{i\in\mathcal{S}^C}H_{\mathcal{S}^C}H_{\mathcal{S}^C}^*\right).
\end{multline}
Since the channel is ergodic, the minimum in (\ref{eq:exp1}) is over
functionals of the channel probability, rather then channel
realizations. In addition, the channel probability is symmetric with
regards to the agents, leading to $F_{\{1,\dots,r\}}(q^*)$, which is
the minimum among all the subsets $\mathcal{S}$. So that the
achievable rate can be calculated by solving the following problem
\begin{equation}
\max_{\{q_i:\mathbb{C}^{[r\times
t]}\rightarrow\mathbb{R}_{+}\}_{i=1}^r} \mathrm{E}_H
\log_2\det\left(I_2+\frac{P}{t}\mathrm{diag}\left(1-2^{-q_i(H)}\right)_{i=1}^rHH^*\right)
\end{equation}
such that $q_i(H)\geq 0$ and
\begin{equation}
\mathrm{E}_H[q_i(H)]=q^*,\ i=1,\dots,r.
\end{equation}
Let us limit the discussion to the case of $r=2$. The solution can
be obtained through Lagrange multipliers, as follows ($\theta\geq
0$)
\begin{eqnarray}
\triangledown\log_2\det\left(I_2+\frac{P}{t}\mathrm{diag}\left(1-2^{-q_i(H)}\right)_{i=1,2}HH^*\right)-\theta
I_2=\mu(H).
\end{eqnarray}
So for any $\mu_i(H)=0$, such that $q_i(H)>0$, we get that
($\bar{i}=3-i$)
\begin{equation}\label{eq:explic_2}
i=1,2:\
\frac{2^{-q_i}(\Delta_{2+i}-2^{-q_{\bar{i}}}\Delta_2)}{\Delta}=\theta,
\end{equation}
and $\mathrm{E}_H(q_i) = q^*$, where
\begin{eqnarray}
\Delta & \triangleq &
\det\left(I_2+\frac{P}{t}\mathrm{diag}\left(1-2^{-q_i(H)}\right)_{i=1,2}HH^*\right)\\
\Delta_1 & \triangleq & \det\left(I_2+\frac{P}{t}HH^*\right)\\
\Delta_2 & \triangleq & \det\left(\frac{P}{t}HH^*\right)\\
\Delta_3 & \triangleq & \det\left(\mathrm{diag}([0,1])+\frac{P}{t}HH^*\right)\\
\Delta_4 & \triangleq &
\det\left(\mathrm{diag}([1,0])+\frac{P}{t}HH^*\right).
\end{eqnarray}
We note that (\ref{eq:explic_2}) determines a one-to-one
connection between $\theta$ and $q^*$. In addition, note that
\begin{equation*}
\Delta = \Delta_1
+2^{-q_1-q_2}\Delta_2-2^{-q_1}\Delta_3-2^{-q_2}\Delta_4,
\end{equation*}
and that
\begin{eqnarray}
\Delta_3 = \Delta_2 + \frac{P}{t}|H_1|^2\\
\Delta_4 = \Delta_2 + \frac{P}{t}|H_2|^2.
\end{eqnarray}

The solution of (\ref{eq:explic_2}) is
\begin{equation}\label{eq:explic_4}
q_i=-\log_2\left(\frac{\Delta_{\bar{i}+2}}{2\Delta_2(1+\theta)}\left((1+2\theta)-\sqrt{(1+2\theta)^2-4\theta(1+\theta)\frac{\Delta_1\Delta_2}{\Delta_3\Delta_4}}\right)\right).
\end{equation}
We note that $\frac{\Delta_1\Delta_2}{\Delta_3\Delta_4}\leq 1$
with equality if and only if $HH^*$ is a diagonal matrix. So the
square root in equation (\ref{eq:explic_4}) is guaranteed to be
positive real. By a simple derivative, it is easily verified that
$F_H(\theta)$, defined by (\ref{eq:explic_F}), is monotonically
increasing with $\theta$.

Then, in case any of $q_i,\ i=1,2$ from (\ref{eq:explic_4}) turns
out negative (say $F_H(\theta)>\frac{\Delta_2}{\Delta_{2+i}}$
which leads to $q_{\bar{i}}<0$), then the solution is
$q_{\bar{i}}=0$ and $q_i$ is equal to
\begin{equation}\label{eq:explic_5}
q_i=-\log_2\left(\frac{\theta}{1+\theta}\frac{1+\frac{P}{t}|H_i|^2}{\frac{P}{t}|H_i|^2}\right).
\end{equation}
If (\ref{eq:explic_5}) is negative as well, the solution is
$q_i=0$. As $\theta$ gets smaller, more channels will result with
(\ref{eq:explic_4}) solved with $q_i> 0$, which means better
compression.

Overall, we can write
\begin{equation}\label{eq:explic_6}
q_1(H,\theta)=\left\{\begin{array}{cc}
\left\lceil-\log_2\left(\frac{\theta}{1+\theta}\frac{1+\frac{P}{t}|H_1|^2}{\frac{P}{t}|H_1|^2}\right)\right\rceil^+
& F_H(\theta)>\frac{\Delta_2}{\Delta_3} \\
\left\lceil-\log_2(\frac{\Delta_{4}}{\Delta_2}F_H(\theta))\right\rceil^+
& F_H(\theta)\leq \frac{\Delta_2}{\Delta_3}.\end{array}\right.
\end{equation}

Now $\theta$ is determined by the equation
\begin{equation}\label{eq:explic_7}
\mathrm{E}_H
\log_2\det\left(I_2+\frac{P}{t}\mathrm{diag}\left(1-2^{-q_i(H,\theta)}\right)_{i=1}^2HH^*\right)=2(C-\mathrm{E}_H
[q_i(H,\theta)])
\end{equation}
and the achievable rate is
\begin{equation}\label{eq:explic_8}
R_{CEO}=2(C-\mathrm{E}_H [q_i(H,\theta)]).
\end{equation}
This concludes the proof. \hfill{\QED}

\section{Proof of Lemma \ref{lem:EPI_1}}
The proof is divided into two sections, we start by proving for the
case where $|\mathcal{S}|\leq t$. This division is since the first
case is easier to show, and thus gives better understanding of the
guidelines and techniques, which are identical, albeit more
involved, for the case of $|\mathcal{S}|\geq t$.

For the sake of the proof, define:
\begin{itemize}
\item $Z\triangleq H_{\mathcal{S}}X$, where $I(\boldsymbol{Y}_{\mathcal{S}};\boldsymbol{X}|\boldsymbol{H})=I(\boldsymbol{Y}_{\mathcal{S}};\boldsymbol{Z}|\boldsymbol{H})$.
\item $\Lambda_z\triangleq \mathrm{E}[ZZ^*]=H_{\mathcal{S}}QH_{\mathcal{S}}^*=\frac{P}{t}H_{\mathcal{S}}H_{\mathcal{S}}^*$ (equal to $\Lambda_\mathcal{S}$).
\item $\hat{Z}\triangleq AY$, where $A$ is the best estimator of
$Z$ from $Y$, calculated as $A=\Lambda_z(I+\Lambda_z)^{-1}$.
\end{itemize}
Since $|\mathcal{S}|\leq t$ we have that $|\Lambda_z|>0$. Note
that since $\hat{Z}$ is the best estimator
\begin{equation}\label{eq:indi_11}
Z=\hat{Z}+\hat{N}
\end{equation}
where $\hat{Z}$ and $\hat{N}$ are independent, and since
$\mathrm{E}[\hat{Z}\hat{Z}^*]=\Lambda_z(I+\Lambda_z)^{-1}\Lambda_z$,
we get $\mathrm{E}[\hat{N}\hat{N}^*]=\Lambda_z(I+\Lambda_z)^{-1}$.
Now we can rely on the independence in (\ref{eq:indi_11}) and the
vector entropy power inequality:
\begin{equation}\label{eq:EPI_11}
2^{\frac{1}{n|\mathcal{S}|}h(\boldsymbol{Z}|V_\mathcal{S},\boldsymbol{H}=\boldsymbol{h})}\geq
2^{\frac{1}{n|\mathcal{S}|}h(\boldsymbol{\hat{Z}}|V_\mathcal{S},\boldsymbol{H}=\boldsymbol{h})}+(\pi
e)\prod_{k=1}^n
\left(\frac{|\Lambda_z(k)|}{|I+\Lambda_z(k)|}\right)^{\frac{1}{n|\mathcal{S}|}}.
\end{equation}
Next we express the required quantity $\lambda\triangleq
\frac{1}{n}I(\boldsymbol{Z};V_\mathcal{S}|\boldsymbol{H}=\boldsymbol{h})$
in both sides of (\ref{eq:EPI_11}). For the left hand side,
\begin{equation}\label{eq:EPI_12}
\frac{1}{n|\mathcal{S}|}h(\boldsymbol{Z}|V_\mathcal{S},\boldsymbol{H}=\boldsymbol{h})=\frac{1}{n|\mathcal{S}|}h(\boldsymbol{Z}|\boldsymbol{H}=\boldsymbol{h})-\frac{\lambda}{|\mathcal{S}|}=\frac{1}{n|\mathcal{S}|}\log_2\left(\prod^n|\Lambda_z(k)|\right)+\log_2(\pi
e )-\frac{\lambda}{|\mathcal{S}|}.
\end{equation}
The right hand side is more elaborated, and will be done in two
stages. First note that:
\begin{equation}\label{eq:EPI_13}
h(\boldsymbol{\hat{Z}}|V_\mathcal{S},\boldsymbol{H}=\boldsymbol{h})=h(\boldsymbol{\hat{Z}}|\boldsymbol{Z},V_\mathcal{S},\boldsymbol{H}=\boldsymbol{h})+I(\boldsymbol{Z};\boldsymbol{\hat{Z}}|V_\mathcal{S},\boldsymbol{H}=\boldsymbol{h}).
\end{equation}
We know that
$h(\boldsymbol{Z}|\boldsymbol{\hat{Z}},\boldsymbol{H}=\boldsymbol{h})=h(\boldsymbol{Z}|\boldsymbol{\hat{Z}},V_\mathcal{S},\boldsymbol{H}=\boldsymbol{h})$,
from the definition of $\hat{Z}$ and $V_\mathcal{S}$. This means
that:
\begin{equation}\label{eq:EPI_14}
\frac{1}{n}I(\boldsymbol{Z};\boldsymbol{\hat{Z}}|V_\mathcal{S},\boldsymbol{H}=\boldsymbol{h})=\frac{1}{n}I(\boldsymbol{Z};\boldsymbol{\hat{Z}}|\boldsymbol{H}=\boldsymbol{h})-\lambda=\frac{1}{n}\log_2\left(\prod^n|I+\Lambda_z(k)|\right)-\lambda.
\end{equation}
Second, we have that $\hat{Z}=AY_\mathcal{S}$, so
\begin{multline}\label{eq:EPI_15}
h(\boldsymbol{\hat{Z}}|\boldsymbol{Z},V_\mathcal{S},\boldsymbol{H}=\boldsymbol{h})=h(\boldsymbol{Y}_\mathcal{S}|\boldsymbol{Z},V_\mathcal{S},\boldsymbol{H}=\boldsymbol{h})+2\log_2\left(\prod^n|A(k)|\right)=\\
\sum_{i\in\mathcal{S}}h(\boldsymbol{Y}_i|\boldsymbol{Z},V_i,\boldsymbol{H}=\boldsymbol{h})+2\log_2\left(\prod^n\frac{|\Lambda_z(k)|}{|I+\Lambda_z(k)|}\right).
\end{multline}
define
$q_i(\boldsymbol{h})\triangleq\frac{1}{n}I(\boldsymbol{Y}_i;V_i|\boldsymbol{X},\boldsymbol{H}=\boldsymbol{h})$
and since we used additive noise with unit variance,
\begin{equation}\label{eq:EPI_16}
\frac{1}{n}h(\boldsymbol{Y}_i|\boldsymbol{Z},V_i,\boldsymbol{H}=\boldsymbol{h})=\log_2(\pi
e)-q_i(\boldsymbol{h}).
\end{equation}
rewrite (\ref{eq:EPI_15}) as
\begin{equation}\label{eq:EPI_17}
\frac{1}{n}h(\boldsymbol{\hat{Z}}|\boldsymbol{Z},V_\mathcal{S},\boldsymbol{H}=\boldsymbol{h})=|\mathcal{S}|\log_2(\pi
e)-\sum_{i\in\mathcal{S}}q_i(\boldsymbol{h})+\frac{2}{n}\log_2\left(\prod^n\frac{|\Lambda_z(k)|}{|I+\Lambda_z(k)|}\right).
\end{equation}
Now using (\ref{eq:EPI_14}) and (\ref{eq:EPI_17}) in the right
hand side, written in (\ref{eq:EPI_13}), we get to:
\begin{equation}\label{eq:EPI_18}
2^{\frac{1}{n|\mathcal{S}|}h(\boldsymbol{\hat{Z}}|V_\mathcal{S},\boldsymbol{H}=\boldsymbol{h})}=\pi
e
\prod_{i\in\mathcal{S}}2^{-\frac{q_i(\boldsymbol{h})}{|\mathcal{S}|}}\left(\prod^n\frac{|\Lambda_z(k)|}{|I+\Lambda_z(k)|}\right)^{\frac{2}{n|\mathcal{S}|}}
\left(\prod^n|I+\Lambda_z(k)|\right)^{\frac{1}{n|\mathcal{S}|}}2^{-\frac{\lambda}{|\mathcal{S}|}}.
\end{equation}
Finally we combine left hand side (\ref{eq:EPI_12}) and right hand
side (\ref{eq:EPI_18}), and get
\begin{equation}\label{eq:EPI_19}
\pi
e2^{-\frac{1}{|\mathcal{S}|}\lambda}\left(\prod^n|\Lambda_z(k)|\right)^{\frac{1}{n|\mathcal{S}|}}
 \geq \pi e 2^{-\frac{\lambda}{|\mathcal{S}|}}
\prod_{i\in\mathcal{S}}2^{-\frac{q_i(\boldsymbol{h})}{|\mathcal{S}|}}\left(\prod^n\frac{|\Lambda_z(k)|}{|I+\Lambda_z(k)|}\right)^{\frac{2}{n|\mathcal{S}|}}
\left(\prod^n|I+\Lambda_z(k)|\right)^{\frac{1}{n|\mathcal{S}|}}+\pi
e\prod_{k=1}^n
\left(\frac{|\Lambda_z(k)|}{|I+\Lambda_z(k)|}\right)^{\frac{1}{n|\mathcal{S}|}}.
\end{equation}
Reordering the equation we get to (\ref{eq:upper_1}), which proves
Lemma \ref{lem:EPI_1} for when
$|\mathcal{S}|\leq t$.\\

We continue to the case where $|\mathcal{S}|>t$, where we have more
agents than transmitters, so that $|\Lambda_z|=0$. Like in the
previous setting we define $\hat{X}=AY$ to be the best estimator of
$X$ out of $Y$. So that now $A=QH^*(I+\Lambda_z)^{-1}$, and we have
\begin{equation}
X=\hat{X}+\hat{N},
\end{equation}
where $\hat{X}$ and $\hat{N}$ are independent and using the matrix
inversion Lemma
$\mathrm{E}[\hat{N}\hat{N}^*]=(Q^{-1}+H^*H)^{-1}=Q(I+QH^*H)^{-1}$.
Again we use the entropy power inequality:
\begin{equation}\label{eq:EPI_21}
2^{\frac{1}{nt}h(\boldsymbol{X}|V_\mathcal{S},\boldsymbol{H}=\boldsymbol{h})}\geq
2^{\frac{1}{nt}h(\boldsymbol{\hat{X}}|V\mathcal{S},\boldsymbol{H}=\boldsymbol{h})}+\pi
e |Q|^{\frac{1}{t}}\prod_{k=1}^n
\left(\frac{1}{|I+\Lambda_z(k)|}\right)^{\frac{1}{nt}}.
\end{equation}
Using the same argument as the one used for (\ref{eq:EPI_12}), the
left hand side of (\ref{eq:EPI_21}) becomes
\begin{equation}\label{eq:EPI_22}
2^{\frac{1}{nt}h(\boldsymbol{X}|V_\mathcal{S},\boldsymbol{H}=\boldsymbol{h})}=\pi
e |Q|^{\frac{1}{t}} 2^{-\frac{\lambda}{t}}.
\end{equation}
The left expression in the right hand side of (\ref{eq:EPI_21})
can be written as the sum of two arguments, as in
(\ref{eq:EPI_13}), where the right-most mutual information (like
(\ref{eq:EPI_14})) is
\begin{equation}\label{eq:EPI_24}
\frac{1}{n}I(\boldsymbol{X};\boldsymbol{\hat{X}}|V_\mathcal{S},\boldsymbol{H}=\boldsymbol{h})=\frac{1}{n}I(\boldsymbol{X};\boldsymbol{\hat{X}}|\boldsymbol{H}=\boldsymbol{h})-\lambda=\frac{1}{n}\log_2\left(\prod^n|I+\Lambda_z(k)|\right)-\lambda.
\end{equation}
The difference between the case where $|\mathcal{S}|<t$ and
$|\mathcal{S}|>t$ is evident in the derivation of (\ref{eq:EPI_15}),
which for $|\mathcal{S}|>t$ requires the double use of the entropy
power inequality. So we want to lower bound
$h(\boldsymbol{\hat{X}}|V_\mathcal{S},\boldsymbol{X},\boldsymbol{H}=\boldsymbol{h})$.
First, let us decompose $A$ using the singular value decomposition
into $A=U_1DU_2$, where $U_1\in\mathbb{C}^{t\times t}$ and
$U_2\in\mathbb{C}^{|\mathcal{S}|\times |\mathcal{S}|}$ are two
unitary matrices and $D\in\mathrm{R}^{t\times |\mathcal{S}|}$ is
diagonal matrix. So we have that:
\begin{multline}\label{eq:EPI_25}
h(\boldsymbol{\hat{X}}|V_\mathcal{S},\boldsymbol{X},\boldsymbol{H}=\boldsymbol{h})=h(\boldsymbol{U}_1\boldsymbol{D}\boldsymbol{U}_2\boldsymbol{Y}|V_\mathcal{S},\boldsymbol{X},\boldsymbol{H}=\boldsymbol{h})\\=\log_2\prod^n|U_1(k)|^2+\sum_{j=1}^t[
\log_2\prod^n|D_{j,j}(k)|^2+h((\boldsymbol{U}_2)_j\boldsymbol{Y}|V_\mathcal{S},\boldsymbol{X},\boldsymbol{H}=\boldsymbol{h})],
\end{multline}
since $U_2$ is unitary matrix. Next we employ the entropy power
inequality to lower bound
$h((\boldsymbol{U}_2)_j\boldsymbol{Y}|V_\mathcal{S},\boldsymbol{X},\boldsymbol{H}=\boldsymbol{h})$:
\begin{equation}
2^{h((\boldsymbol{U}_2)_j\boldsymbol{Y}|V_\mathcal{S},\boldsymbol{X},\boldsymbol{H}=\boldsymbol{h})}\geq
\sum_{i\in\mathcal{S}}
2^{h(\boldsymbol{Y}_i|V_\mathcal{S},\boldsymbol{X},\boldsymbol{H}=\boldsymbol{h})}\prod^n|(U_2(k))_{j,i}|^2.
\end{equation}
This inequality is achieved with equality for Gaussian variables. A
lower bound on (\ref{eq:EPI_25}) is given by
\begin{multline}\label{eq:EPI_27}
h(\boldsymbol{\hat{X}}|V_\mathcal{S},\boldsymbol{X},\boldsymbol{H}=\boldsymbol{h})\geq\log_2\left(\prod^n|U_1(k)D(k)U_2(k)\mathrm{diag}(2^{-q_i(\boldsymbol{h})})_{i\in\mathcal{S}}U_2(k)^*D(k)^*U_1(k)^*|\right)+nt\log_2(\pi e)\\
=\log_2\left(\prod^n|QH(k)^*(I+\Lambda_z(k))^{-1}\mathrm{diag}(2^{-q_i(\boldsymbol{h})})_{i\in\mathcal{S}}(I+\Lambda_z(k))^{-1}H(k)Q|\right)+nt\log_2(\pi
e)\\
=\log_2\left(\prod^n\left(\frac{|QH(k)^*\mathrm{diag}(2^{-q_i(\boldsymbol{h})})_{i\in\mathcal{S}}H(k)Q|}{|I+\Lambda_z(k)|^2}\right)\right)+nt\log_2(\pi
e)
\end{multline}
since
\begin{equation}
|QH^*(I+\Lambda_z)^{-1}D(I+\Lambda_z)^{-1}HQ|=\frac{|(I+QH^*H)QH^*(I+\Lambda_z)^{-1}D(I+\Lambda_z)^{-1}HQ(I+H^*HQ)|}{|I+\Lambda_z|^2}=\frac{|QH^*DHQ|}{|I+\Lambda_z|^2}.
\end{equation}
To conclude, we use (\ref{eq:EPI_22}), (\ref{eq:EPI_24}) and
(\ref{eq:EPI_27}):
\begin{equation}
\pi e |Q|^{\frac{1}{t}} 2^{-\frac{\lambda}{t}}\geq\pi
e2^{-\frac{-\lambda}{t}}\left(\prod^n|I+\Lambda_z(k)|\frac{|Q|^2|H(k)^*\mathrm{diag}(2^{-q_i(\boldsymbol{h})})_{i\in\mathcal{S}}H(k)|}{|I+\Lambda_z(k)|^2}\right)^{\frac{1}{nt}}+\pi
e |Q|^{\frac{1}{t}}\prod_{k=1}^n
\left(\frac{1}{|I+\Lambda_z(k)|}\right)^{\frac{1}{nt}}
\end{equation}
which by taking expectation with respect to $\boldsymbol{H}$,
together with (\ref{eq:EPI_19}) proves Lemma \ref{lem:EPI_1}. \QED
\section{Proof of Proposition \ref{prop:achv_BC}}
The proof of Proposition \ref{prop:achv_BC} is based on the proof of
Theorem 3 from \cite{fullpaper}.
\subsection{Code construction:}\label{subsec:Code_construction}
For every channel realization $\boldsymbol{h}$, determine the
maximizing $\pi$. Fix $\delta>0$ and then
\begin{enumerate}
\item\label{codecnst:BC} For the broadcast transmissions, for every
$i=\pi_1,\dots,\pi_r$:
\begin{itemize}
\item Randomly generate $2^{n [I(W_{i};W_{\tilde{\mathcal{T}}(\pi,i)}|H)+\delta]}$
vectors $\boldsymbol{w}_i$, according to
$P_{\boldsymbol{W}_{i}|\boldsymbol{H}}(\boldsymbol{w}_{i}|\boldsymbol{h})=\prod_{k=1}^n
P_{W_{i}|H}(w_{i}(k)|h(k))$.
\item For every $\boldsymbol{w}_{\tilde{\mathcal{T}}(\pi,i)}$
generated in the previous iteration, find at least one
$\boldsymbol{w}_i$ within the generated set which is jointly
typical. Joint typicality means that
\begin{equation}
(\boldsymbol{w}_i,\boldsymbol{w}_{\tilde{\mathcal{T}}(\pi,i)})\in
\mathbf{T}_{\epsilon}^{BC,i}(\boldsymbol{h}),
\end{equation}
where
\begin{multline}
\mathbf{T}^{BC,i}_\epsilon(\boldsymbol{h}) \triangleq\
\Bigg\{\boldsymbol{w}_{i,\tilde{\mathcal{T}}(\pi,i)}:\ \forall
\mathcal{S}\subseteq\{i,\tilde{\mathcal{T}}(\pi,i)\},\ \forall
w_\mathcal{S} \in
\mathcal{W}_\mathcal{S},h\in \mathcal{H}\\
\frac{1}{n}\left|N(w_\mathcal{S},h|\boldsymbol{w}_\mathcal{S},\boldsymbol{h})-P_{W_\mathcal{S}|H}(w_\mathcal{S}|\boldsymbol{h})N(h|\boldsymbol{h})\right|<
\frac{\epsilon}{ |\mathcal{W}_\mathcal{S}|} \Bigg\}.
\end{multline}
\item In case no such vector exists, declare error event $E_1$.
\item Repeat the last steps for $2^{n [I(W_{i};Y_{i}|H)-I(W_{i};W_{\tilde{\mathcal{T}}(\pi,i)}|H)-\delta]}$
times.
\end{itemize}
Label the resulting vectors of each repetition, which were jointly
typical, by $M_{i}$,\\ where $M_{i}\in[1,2^{n
[I(W_{i};Y_{i}|H)-I(W_{i};W_{\tilde{\mathcal{T}}(\pi,i)}|H)-\delta]}]$.
Then $M^r=\{M_1,\dots,M_r\}$ and further define
$\mathcal{M}_{M_{i}}$ as the set labeled by $M_{i}$. So that
$\boldsymbol{w}^r(M^r,\boldsymbol{h})$ are the $r$ vectors which
were selected in the last stage and are jointly typical.
\item\label{codecnst:CF_tx} For the message which is decoded at the
final destination, for every $\boldsymbol{w}^r$ defined by some
$M^r$, and for every random encoding realization $f$
\begin{itemize}
\item Randomly choose $2^{nR_{CEO}}$ vectors $\boldsymbol{x}$, with probability
$P_{\boldsymbol{X}|\boldsymbol{W}^r,\boldsymbol{H}}(\boldsymbol{x}|\boldsymbol{w}^r,\boldsymbol{h})=\prod_k
P_{X|W^r,H}(x(k)|w^r(k),h(k))$.
\item Index these vectors by $M_{CEO}$ where $M_{CEO}\in[1,2^{nR_{CEO}}]$.
\item So we have $2^{n[\sum_{i=1}^r I(W_i;Y_i|H)-I(W_{i};W_{\tilde{\mathcal{T}}(\pi,i)}|H)-\delta]}$ different mappings between
indices $M_{CEO}$ and vectors $\boldsymbol{x}$, where the one used
is determined by $M^r$. We will therefore denote
$\boldsymbol{x}(M_{CEO},M^r,\boldsymbol{h})$ as the vector indexed
by $M_{CEO},M^r$. We leave out the notation of $f$ in the sequel,
for the sake of brevity, since for decoding agents, the chosen $f$
is known at the agents, so the achievable rate is valid for every
realization of $f$, with high probability.
\end{itemize}
\item\label{codecnst:agents} For the compressor at the agents\\
For all $\boldsymbol{w}^r$ indicated by $M^r$,
\begin{itemize}
\item Randomly generate $2^{n[\hat{R}_i-(C_i-\{I(W_{i};Y_{i}|H)-I(W_{i};W_{\tilde{\mathcal{T}}(\pi,i)}|H)-\delta\})]}$ vectors $\boldsymbol{u}_i$ of length $n$\\ according to
$\prod_k
P_{U_i|W_i,H}(u_{i}(k)|w_{i}(k),h(k))$. 
\item Repeat the last step for $s_i=1,\dots,2^{n(C_i-\{I(W_{i};Y_{i}|H)-I(W_{i};W_{\tilde{\mathcal{T}}(\pi,i)}|H)-\delta\})}$, define the resulting set of
$\boldsymbol{u}_i$ of each repetition by $S_{s_i}$.
\item Index all the generated $\boldsymbol{u}_i$ with $z_i\in[1,2^{n\hat{R}_i}]$.
We will interchangeably use the notation $S_{s_i}$ for the set of
vectors $\boldsymbol{u}_i$ as well as for the set of the
corresponding $z_i$.
\item Notice that the mapping between
the indices $z_i$ and the vectors $\boldsymbol{u}_i$ depends on
$\boldsymbol{w}_i,\boldsymbol{h}$. So we will write
$\boldsymbol{u}_i(z_i,\boldsymbol{w}_i,\boldsymbol{h})$ to denote
$\boldsymbol{u}_i$ which is indexed by $z_i$ for some specific
$\boldsymbol{w}_i,\boldsymbol{h}$.
\end{itemize}
\end{enumerate}
\subsection{Encoding:}\label{subsec:encoding}

Let $M=(M^r,M_{CEO})$ be the message to be sent ($M^r$ is defined at
the previous subsection), and the channel realizations be
$\boldsymbol{h}$. The transmitter then sends
$\boldsymbol{x}(M_{CEO},M^r,\boldsymbol{h})$ to the channel.
\subsection{Processing at the
agents:}\label{subsec:agent_processing}
\subsubsection{Decoding}\label{subsubsec:decodeing}
The $i^{th}$ agent knows $\boldsymbol{h}$ and receives
$\boldsymbol{y}_i$ from the channel. It looks for
$\boldsymbol{\hat{w}}_i$ so that
\begin{equation}
(\boldsymbol{y}_i,\boldsymbol{\hat{w}}_i)\in
\mathbf{T}_\epsilon^{i,1}(\boldsymbol{h}),
\end{equation}
where
\begin{equation}
\mathbf{T}^{i,1}_\epsilon(\boldsymbol{h}) \triangleq\
\left\{\boldsymbol{w}_i,\boldsymbol{y}_i:\ \begin{array}{c}\forall w
\in
\mathcal{W}_i,h\in \mathcal{H}:\ 
\frac{1}{n}\left|N(w,h|\boldsymbol{w}_i,\boldsymbol{h})-P_{W_i|H}(w|h)N(h|\boldsymbol{h})\right|<
\frac{\epsilon}{ |\mathcal{W}_i|} \\
\forall y \in \mathcal{Y}_i,h\in \mathcal{H}:\ 
\frac{1}{n}\left|N(y,h|\boldsymbol{y}_i,\boldsymbol{h})-P_{Y_i|H}(y|h)N(h|\boldsymbol{h})\right|<
\frac{\epsilon}{ |\mathcal{Y}_i|}
\\
w \in
\mathcal{W}_i,y \in \mathcal{Y}_i,h\in \mathcal{H}:\ 
\frac{1}{n}\left|N(w,y,h|\boldsymbol{w}_i,\boldsymbol{y}_i,\boldsymbol{h})-P_{W_i,Y_i|H}(w,y|h)N(h|\boldsymbol{h})\right|<
\frac{\epsilon}{ |\mathcal{Y}_i||\mathcal{W}_i|}
\end{array}\right\}.
\end{equation}

If no such $\boldsymbol{\hat{w}}_i$ exists, chose arbitrary
$\boldsymbol{\hat{w}}_i$, and if more than one is found, select one
of them arbitrarily. Denote by $E_2$ the error event where the
chosen vector $\boldsymbol{\hat{w}}_i\neq
\boldsymbol{w}_i(M^r,\boldsymbol{h})$.

\subsubsection{Compression}\label{subsub:compression}
The $i^{th}$ agent chooses any of the $z_i$ such that
\begin{equation}\label{eq:relay_z_typical2}
\big(\boldsymbol{u}_i(z_i,\boldsymbol{\hat{w}}_i,\boldsymbol{h}),\boldsymbol{y}_i,\boldsymbol{\hat{w}}_i\big)\in\mathbf{T}_{\epsilon}^{t,2}(\boldsymbol{h}).
\end{equation}
The event where no such $z_i$ is found is defined as the error
event $E_3$.\\
After deciding on $z_i$ the agent transmits $s_i$, which fulfills
$z_i\in S_{s_i}$, and $\hat{M}_i$ to the final destination through
the lossless link, where $\hat{M}_i$ corresponds to
$\boldsymbol{\hat{w}}_t$.
\subsection{Decoding (at the destination):}\label{subsec:Decoding (at the destination)}
The destination retrieves $\hat{M}^r$ and $s^r\triangleq (s_1,\dots,s_r)$ from the lossless links.\\
The destination then finds the set of indices
$\hat{z}^r\triangleq\{\hat{z}_1,\dots,\hat{z}_r\}$ of the compressed
vectors $\boldsymbol{\hat{u}}^r$ and $\hat{M_{CEO}}$ which satisfy
\begin{equation}
\begin{cases}
\big(\boldsymbol{x}(\hat{M}_{CEO},\hat{M}^r,\boldsymbol{h},f),\boldsymbol{\hat{u}}^r(\hat{z}^r,\hat{M}^r,\boldsymbol{h}),\boldsymbol{\hat{w}}^r(\hat{M}^r,\boldsymbol{h})\big)\in
\mathbf{T}_{\epsilon}^3(\boldsymbol{h})\\
\hat{z}^r\in S_{s_1}\times\dots\times S_{s_r}.
\end{cases}
\end{equation}
Where $\mathbf{T}_{\epsilon}^3$ is defined in the standard way, as
(\ref{eq:typical_set_def}). If there is no such
$\hat{z}^r,\hat{M}_{CEO}$, or if there is more than one, the
destination chooses one arbitrarily. Define error $E_4$ as the event
where $\hat{M}_{CEO}\neq
M_{CEO}$.\\
Correct decoding means that the destination decides $\hat{M}=M$.
An achievable rate $R$ was defined as when the final destination
receives the transmitted message with an error probability which
is made arbitrarily small for sufficiently large block length $n$.
\subsection{Error analysis}\label{subsec:err_ana2}
The error probability is upper bounded by:
\begin{equation}
\Pr\{\mathrm{error}\}=\Pr\left(\cup_{i=1}^{4} E_i\right)\leq
\sum_{i=1}^{4}\Pr(E_i).
\end{equation}
Where:
\begin{enumerate}
\item $E_1$: No $r$-tuple $\boldsymbol{w}^r$ jointly typical is
found.
\item $E_2$: A different $\boldsymbol{\hat{w}}_i\neq \boldsymbol{w}_i$ is selected by the $i^{th}$ agent.
\item $E_3$: No $\boldsymbol{u}_i(z_i,\boldsymbol{\hat{w}}_i,\boldsymbol{h})$ is jointly typical with
$(\boldsymbol{y}_i,\boldsymbol{\hat{w}}_i)$.
\item $E_4$: Decoding error $\boldsymbol{x}(\hat{M}_{CEO},\hat{M}_\mathcal{T},f)\neq \boldsymbol{x}(M,f)$.
\end{enumerate}
Next, we will upper bound the probabilities of the individual error
events by arbitrarily small $\epsilon$.
\subsubsection{$E_1$}
From Lemma \ref{lem:gen_Markov_lem}, it is evident that $\Pr(E_1)$
can be made as small as desired, when $n$ is increased, as long as
$\delta>0$.
\subsubsection{$E_2$}
By Lemma \ref{lem:AEP}, the probability of jointly distributed
variables not to be $\epsilon$-typical is as small as desired for
$n$ sufficiently large. According to Lemma
\ref{lem:jointly_typical}, the probability that another
$\boldsymbol{\hat{w}}_i$ belongs to $\mathbf{T}_\epsilon^{i,1}$ is
upper bounded by $2^{-n[I(W_i;Y_i|H)-\epsilon]}$. Since there are no
more than
$2^{n[I(W_i;Y_i|H)-I(W_{i};W_{\tilde{\mathcal{T}}(\pi,i)}|H)-\delta]}$
such $\boldsymbol{\hat{w}}_i$, the probability of $E_2$ can be made
arbitrarily small as $n$ goes to infinity as long as
$I(W_{i};W_{\tilde{\mathcal{T}}(\pi,i)}|H)+\delta>\epsilon$.
\subsubsection{$E_3$}
According to Lemma \ref{lem:gen_Markov_lem}, the probability
$\Pr\{E_3\}$ can be made as small as desired, for $n$ sufficiently
large, as long as
\begin{equation}\label{eq:single_compression_rate}
\hat{R}_i>I(U_i;Y_i|W_i,H).
\end{equation}
\subsubsection{$E_4$}
Consider the case where $\hat{M}_{CEO}\neq M_{CEO}$ and
$\hat{z}_\mathcal{S}\neq z_\mathcal{S}$. There are
$$2^{n[R_{CEO}+\sum_{i\in\mathcal{S}}[\hat{R}_i-(C_i-\{I(W_{i};Y_{i}|H)-I(W_{i};W_{\tilde{\mathcal{T}}(\pi,i)}|H)-\delta\})]]}$$
such vectors, and the probability of
$(\boldsymbol{x}(\hat{M}),\boldsymbol{u}_\mathcal{S}(\hat{z}_\mathcal{S}),\boldsymbol{u}_{\mathcal{S}^C}(\hat{z}_{\mathcal{S}^C}))$
to be jointly typical is upper bounded by (Lemma
\ref{lem:jointly_typical})
$2^{n[H(X,U^r|W^r,H)-H(X|W^r,H)-H(U_{\mathcal{S}^C}|W^r,H)-\sum_{i\in\mathcal{S}}H(U_i|W_i,H)+\epsilon]}$.
Thus the rate $R_{CEO}$ is achievable if:
\begin{multline}
R_{CEO}<\sum_{i\in\mathcal{S}}[C_i-\{I(W_{i};Y_{i}|H)-I(W_{i};W_{\tilde{\mathcal{T}}(\pi,i)}|H)-\delta\}-\hat{R}_i+H(U_i|W_i,H)]-H(U_\mathcal{S}|X,W^r,H)-H(U_{\mathcal{S}^C}|X,U_\mathcal{S},W^r,H)\\
<\sum_{i\in\mathcal{S}}[C_i-\{I(W_{i};Y_{i}|H)-I(W_{i};W_{\tilde{\mathcal{T}}(\pi,i)}|H)-\delta\}-I(Y_i;U_i|X,W_i,H)]+I(U_{\mathcal{S}^C};X|W^r,H),
\end{multline}
where the second inequality is due to
(\ref{eq:single_compression_rate}) and because of the Markov chain
$U_i-(W^r,X,H)-U_{1,\dots,i-1,i+1,\dots,r}$. 
Finally, the overall achievable rate is equal to
\begin{equation}
R_{CEO}+\sum_{i=1}^r\{I(W_{i};Y_{i}|H)-I(W_{i};W_{\tilde{\mathcal{T}}(\pi,i)}|H)-\delta\},
\end{equation}
which proves Proposition \ref{prop:achv_BC}. \hfill{\QED}
%

\bibliographystyle{IEEEtran}
\bibliography{IEEEabrv,bounding_agents_capacity}

\begin{thebibliography}{10}
\providecommand{\url}[1]{#1}
\csname url@rmstyle\endcsname
\providecommand{\newblock}{\relax}
\providecommand{\bibinfo}[2]{#2}
\providecommand\BIBentrySTDinterwordspacing{\spaceskip=0pt\relax}
\providecommand\BIBentryALTinterwordstretchfactor{4}
\providecommand\BIBentryALTinterwordspacing{\spaceskip=\fontdimen2\font plus
\BIBentryALTinterwordstretchfactor\fontdimen3\font minus
  \fontdimen4\font\relax}
\providecommand\BIBforeignlanguage[2]{{%
\expandafter\ifx\csname l@#1\endcsname\relax
\typeout{** WARNING: IEEEtran.bst: No hyphenation pattern has been}%
\typeout{** loaded for the language `#1'. Using the pattern for}%
\typeout{** the default language instead.}%
\else
\language=\csname l@#1\endcsname
\fi
#2}}

\bibitem{fullpaper}
A.~Sanderovich, S.~Shamai, Y.~Steinberg, and G.~Kramer, ``Communication via
  decentralized processing,'' \emph{Submitted to IEEE trans. on info. Theory},
  Nov. 2005, partially presented at \cite{Confpaper}.

\bibitem{Telatar99}
E.~Telatar, ``Capacity of multi-antenna {Gaussian} channels,''
  \emph{\textit{European Transactions on Telecommunications}, ETT}, vol.~10,
  no.~6, pp. 585--596, Nov. 1999.

\bibitem{TseZheng2003}
D.~N.~C. Tse and L.~Zheng, ``Diversity and multiplexing: A fundamental tradeoff
  in multiple-antenna channels,'' \emph{{IEEE} Trans. Inform. Theory}, vol.~49,
  no.~50, pp. 1073--1096, May 2003.

\bibitem{AvestimehrTse2006}
A.~S. Avestimehr and D.~N. Tse, ``Outage capacity of the fading relay channel
  in the low {SNR} regime,'' \emph{Submitted to IEEE trans. on info. Theory},
  2006.

\bibitem{YukselErkip2006}
M.~Yuksel and E.~Erkip, ``Cooperative wireless systems: A
  diversity-multiplexing tradeoff perspective,'' \emph{{IEEE} Trans. Inform.
  Theory}, 2006, under review.

\bibitem{WeingartenSteinbergShamai2006}
H.~Weingarten, Y.~Steinberg, and S.~Shamai, ``The capacity region of the
  {Gaussian} multiple-input multiple-output broadcast channel,'' \emph{{IEEE}
  Trans. Inform. Theory}, vol.~52, no.~9, pp. 3936--3964, Sep. 2006.

\bibitem{WangZhangHostMadsen2005}
B.~Wang, J.~Zhang, and A.~H{\o}st-Madsen, ``On the capacity of {MIMO} relay
  channels,'' \emph{{IEEE} Trans. Inform. Theory}, vol.~51, no.~1, pp. 29--43,
  Jan. 2005.

\bibitem{BolcskeiNabar2004}
H.~B\"{o}lcskei and R.~U. Nabar, ``Realizing {MIMO} gains without user
  cooperation in large single-antenna adhoc wireless networks,'' in \emph{Proc.
  of IEEE Int. Symp. Info. Theory (ISIT2004)}, Chicago, IL, June 2004, p.~18.

\bibitem{GuptaKumar2003}
P.~Gupta and P.~R. Kumar, ``Towards an information theory of large networks: an
  achievable rate region,'' \emph{{IEEE} Trans. Inform. Theory}, vol.~49,
  no.~8, pp. 1877--1894, Aug. 2003.

\bibitem{MaricYates2004}
I.~Maric and R.~D. Yates, ``Forwarding strategies for {Gaussian} parallel-relay
  networks,'' in \emph{Proc. of IEEE Int. Symp. Info. Theory (ISIT2004)},
  Chicago, IL, June 2004, p. 269.

\bibitem{SomekhZaidelShamai2004}
O.~Somekh, B.~M. Zaidel, and S.~Shamai, ``Spectral efficiency of joint multiple
  cell-site processors for randomly spread {DS-CDMA} systems,'' in \emph{Proc.
  of IEEE Int. Symp. Info. Theory (ISIT2004)}, Chicago, IL, June 2004, p. 278.

\bibitem{Oohama1998}
Y.~Oohama, ``The rate-distortion function for the quadratic {G}aussian {CEO}
  problem,'' \emph{{IEEE} Trans. Inform. Theory}, vol.~44, no.~3, pp.
  1057--1070, May 1998.

\bibitem{WynerZiv1976}
A.~D. Wyner and J.~Ziv, ``The rate-distortion function for source coding with
  side information at the decoder,'' \emph{{IEEE} Trans. Inform. Theory},
  vol.~22, no.~1, pp. 1--10, Jan 1976.

\bibitem{wagner-2005-}
\BIBentryALTinterwordspacing
A.~B. Wagner, S.~Tavildar, and P.~Viswanath, ``The rate region of the quadratic
  {Gaussian} two-terminal source-coding problem,'' \emph{submitted to IEEE
  trans. on IT}, Feb 2006. [Online]. Available:
  \url{http://www.citebase.org/cgi-bin/citations?id=oai:arXiv.org:cs/0510095}
\BIBentrySTDinterwordspacing

\bibitem{Oohama2005}
Y.~Oohama, ``Rate-distortion theory for {Gaussian} multiterminal source coding
  systems with several side informations at the decoder,'' \emph{{IEEE} Trans.
  Inform. Theory}, vol.~51, no.~7, pp. 2577--2593, July 2005.

\bibitem{DraperWornell2004}
S.~C. Draper and G.~W. Wornell, ``Side information aware coding strategies for
  sensor networks,'' \emph{{IEEE} J. Select. Areas Commun.}, vol.~22, no.~6,
  pp. 966--976, Aug. 2004.

\bibitem{CoverElgamal1979}
T.~M. Cover and A.~A. El-Gamal, ``Capacity theorems for the relay channel,''
  \emph{{IEEE} Trans. Inform. Theory}, vol.~25, no.~5, pp. 572--584, Jan 1979.

\bibitem{KramerGastparGupta2004}
G.~Kramer, M.~Gastpar, and P.~Gupta, ``Information-theoretic multi-hopping for
  relay networks,'' in \emph{International Z\"{u}rich seminar on
  communication}, Swizerland, Feb. 2004.

\bibitem{LiuViswanath2006}
T.~Liu and P.~Viswanath, ``An extremal inequality motivated by multiterminal
  information theoretic problems,'' \emph{Submitted to IEEE trans. on info.
  Theory}, 2006.

\bibitem{WitsenhausenWyner1975}
H.~S. Witsenhausen and A.~D. Wyner, ``a canditional entropy bound for a pair of
  discrete random variables,'' \emph{{IEEE} Trans. Inform. Theory}, vol. IT-21,
  no.~5, pp. 493--501, Sep 1975.

\bibitem{ChayatShamai1989}
N.~Chayat and S.~Shamai, ``Extension of an entropy property for binary input
  memoryless symmetric channels,'' \emph{{IEEE} Trans. Inform. Theory},
  vol.~35, no.~5, pp. 1077--1079, Sep. 1989.

\bibitem{TulinoLozanoVerdu2006}
A.~Tulino, A.~Lozano, and S.~Verd\'{u}, ``Capacity-achieving input covariance
  for single-user multi-antenna channels,'' \emph{{IEEE} Trans. Inform.
  Theory}, vol.~5, no.~1, Jan. 2006.

\bibitem{Confpaper}
A.~Sanderovich, S.~Shamai, Y.~Steinberg, and G.~Kramer, ``Communication via
  decentralized processing,'' in \emph{Proc. of IEEE Int. Symp. Info. Theory
  (ISIT2005)}, Adelaide, Australia, Sep. 2005, pp. 1201--1205.

\bibitem{LiGoldsmith2001}
L.~Li and A.~J. Goldsmith, ``Capacity and optimal resource allocation for
  fading broadcast channels—part {I}: Ergodic capacity,'' \emph{{IEEE} Trans.
  Inform. Theory}, vol.~47, no.~3, pp. 1083--1102, Mar. 2001.

\bibitem{BiglieriProakisShamai1998}
E.~Biglieri, J.~Proakis, and S.~Shamai, ``Fading channels:
  Information-theoretic and communications aspects,'' \emph{{IEEE} Trans.
  Inform. Theory}, vol.~44, no.~6, pp. 2619--2692, Oct. 1998.

\bibitem{CoverThomas}
T.~M. Cover and J.~A. Thomas, \emph{Elements of Information theory}.\hskip 1em
  plus 0.5em minus 0.4em\relax John Wiley \& Sons, Inc., 1991.

\bibitem{KrithivasanPradhan2007}
D.~Krithivasan and S.~S. Pradhan, ``Lattices for distributed source coding:
  Jointly {Gaussian} sources and reconstruction of a linear function,''
  \emph{{IEEE} Trans. Inform. Theory}, July 2007, submitted.

\bibitem{CaireShamai2003}
G.~Caire and S.~Shamai, ``On the achievable throughput of a multiantenna
  {Gaussian} broadcast channel,'' \emph{{IEEE} Trans. Inform. Theory}, vol.~49,
  no.~7, pp. 1691--1706, July 2003.

\bibitem{VishwanathJindalGoldsmith2003}
S.~Vishwanath, N.~Jindal, and A.~Goldsmith, ``The {"Z"} channel,'' in
  \emph{Proceedings of IEEE Global Telecommunications Conference (GlobeCom)},
  San Francisco, CA, Dec.

\bibitem{Csiszar_Korner}
I.~Csisz\'ar and J.~K\"orner, \emph{Information Theory: Coding Theorems for
  Discrete Memoryless Systems}.\hskip 1em plus 0.5em minus 0.4em\relax New
  York: Academic, 1981.

\bibitem{HanKobayashi1980}
T.~S. Han and K.~Kobayashi, ``A unified achievable rate region for a general
  class of multiterminal source coding systems,'' \emph{{IEEE} Trans. Inform.
  Theory}, vol. IT-26, no.~3, pp. 277--288, May 1980.

\end{thebibliography}
\end{document}